\documentclass[journal]{ieeetran}
\ifCLASSINFOpdf
\else
\fi

\usepackage{graphics} 
\usepackage{graphicx}
\usepackage{times} 
\usepackage{amsmath} 
\usepackage{amssymb}  
\usepackage{algorithm}
\usepackage{color}

\def\tr{\text{tr}}

\usepackage{mleftright}

\usepackage{verbatim}

\usepackage{amsthm}

\newcommand{\Dep}{\bm{\Delta}_{\rm prior}}
\newcommand{\Del}{\bm{\Delta}_{\rm learnt}}
\newcommand{\tDecom}{\bm{\tilde{\Delta}}_{\rm com}}
\newcommand{\Decom}{\bm{\Delta}_{\rm com}}
\newcommand{\Det}{\Delta_{\rm tr}}
\newcommand{\tDe}{\tilde{\Delta}}
\newcommand{\tDet}{\tilde{\Delta}_{\rm tr}}
\newcommand{\tDep}{\bm{\tilde{\Delta}}_{\rm prior}}
\newcommand{\tDel}{\bm{\tilde{\Delta}}_{\rm learnt}}

\usepackage{mathtools,bm}

\usepackage{subfigure}

\newtheorem{theorem}{\bf Theorem} \newtheorem{definition}{\bf Definition}
\newtheorem{lemma}{\bf Lemma} \newtheorem{remark}{\bf Remark}
  \newtheorem{proposition}{\bf Proposition}
\newtheorem{assumption}{\bf Assumption} \newtheorem{example}{\bf Example}

\begin{document}
%
\title{Combining Prior Knowledge and Data for Robust Controller Design}
%
%
%

\author{Julian Berberich$^{1}$, Carsten W. Scherer$^{2}$, and Frank Allg\"ower$^{1}$
\thanks{
Funded by Deutsche Forschungsgemeinschaft (DFG, German Research Foundation) under Germany's Excellence Strategy - EXC 2075 - 390740016 and under grant 468094890.
We acknowledge the support by the Stuttgart Center for Simulation Science (SimTech).
The authors thank the International Max Planck Research School for Intelligent Systems (IMPRS-IS) for supporting Julian Berberich.}
\thanks{$^{1}$J. Berberich and F. Allg\"ower are with the University of Stuttgart, Institute for Systems Theory and Automatic Control, 70569 Stuttgart, Germany. {\tt\small E-mail: \{julian.berberich, frank.allgower\}@ist.uni-stuttgart.de}}
\thanks{$^{2}$C. W. Scherer is with the Institute of Mathematical Methods in the Engineering Sciences, Numerical Analysis and Geometrical Modeling, Department of Mathematics, University of Stuttgart, 70569 Stuttgart, Germany. {\tt\small E-mail: carsten.scherer@mathematik.uni-stuttgart.de}}
}

%
%

\markboth{}%
{}
%



\IEEEoverridecommandlockouts

\IEEEpubid{\begin{minipage}{\textwidth}\ \\[12pt] \\ \\
\copyright 2022 IEEE. Personal use of this material is permitted. Permission from IEEE must be obtained for all other uses, in any current or future media, including reprinting/republishing this material for advertising or promotional purposes, creating new collective works, for resale or redistribution to servers or lists, or reuse of any copyrighted component of this work in other works.
\end{minipage}}

\maketitle

\begin{abstract}
We present a framework for systematically combining data of an unknown linear time-invariant system with prior knowledge on the system matrices or on the uncertainty for robust controller design.
Our approach leads to linear matrix inequality (LMI) based feasibility criteria which guarantee stability and performance robustly for all closed-loop systems consistent with the prior knowledge and the available data.
The design procedures rely on a combination of multipliers inferred via prior knowledge and learnt from measured data, where for the latter a novel and unifying disturbance description is employed.
While large parts of the paper focus on linear systems and input-state measurements, we also provide extensions to robust output-feedback design based on noisy input-output data and against nonlinear uncertainties.
We illustrate through numerical examples that our approach provides a flexible framework for simultaneously leveraging prior knowledge and data, thereby reducing conservatism and improving performance significantly if compared to black-box approaches to data-driven control.
\end{abstract}


%
\IEEEpeerreviewmaketitle

\section{Introduction}\label{sec:introduction}
Approaches for controller design based directly on measured data have recently gained increasing attention as they provide many potential benefits if compared to sequential system identification and model-based control.
A key challenge is the development of methods which are simple, i.e., which are less complex than identifying the system and designing a model-based controller, and which provide strong theoretical guarantees, in particular if only finitely many data points are available which may be affected by noise.
However, many existing data-driven control methods are essentially black-box approaches which cannot systematically handle \emph{prior knowledge} about the plant for controller design.
Ultimately, developing tools to merge model-based and data-driven methods in order to simultaneously exploit prior knowledge and data is an important and largely open problem.
In this paper, we present a framework for combining prior knowledge and noisy data of a linear time-invariant (LTI) system for controller design based on robust control theory.

\vskip4pt
{
\noindent
\textit{Related work}}

System identification~\cite{ljung1987system} provides a framework for estimating models from data which can then be used to analyze the system or design a controller.
However, deriving tight error bounds in system identification is a difficult problem in itself and an active field of research even for LTI systems, in particular if non-asymptotic guarantees are desired and the data are perturbed by stochastic noise~\cite{dean2019sample,matni2019self,matni2019tutorial}.
Furthermore, system identification approaches for deterministic noise typically rely on set membership estimation, where providing computationally tractable and tight error bounds from measured data is a key challenge~\cite{belforte1990parameter,milanese1991optimal}.
Thus, exploring alternative approaches for using data directly to design controllers with rigorous end-to-end guarantees is highly interesting and promising, which justifies the
increasing interest in the field~\cite{hou2013model}.
A few selected, established approaches to data-driven control are virtual reference feedback tuning~\cite{campi2002virtual}, unfalsified control~\cite{kosut2001uncertainty}, iterative
feedback tuning~\cite{hjalmarsson1998iterative}, robust control based on frequency domain data~\cite{karimi2017data}, or learning-based model predictive control~\cite{aswani2013provably,berkenkamp2017safe}.
Instead of providing an exhaustive list we refer to~\cite{hou2013model} for an overview of additional existing approaches.

Another recent stream of work which is closely related to the present paper relies on a result from behavioral systems theory.
In~\cite{willems2005note}, it is proven that persistently exciting data can be used directly to parametrize all trajectories of an LTI system, thus providing a promising foundation for data-driven control.
This result has led to the development of various methods for system analysis~\cite{maupong2017lyapunov,romer2019one,koch2021determining,koch2022provably}, state- or output-feedback controller design~\cite{persis2020formulas,persis2021low,berberich2020design,waarde2020informativity,waarde2022from}, model reduction~\cite{monshizadeh2020amidst,burohman2020from,burohman2021from}, internal model control~\cite{rueda2020data}, simulation and optimal control~\cite{markovsky2007on,markovsky2008data}, or predictive control~\cite{yang2015data,coulson2019deepc,berberich2021guarantees,berberich2020constraints,coulson2022distributionally}, all of which are based directly on measured data without any model knowledge.
If compared to many of the existing approaches listed above or in~\cite{hou2013model},
it is a key advancement of those based on~\cite{willems2005note} that they are simple and often come along with strong theoretical guarantees.

Finally, we mention a few selected works on data-driven control which can handle prior knowledge.
First, Gaussian Processes~\cite{rasmussen2006gaussian} permit to incorporate prior knowledge via a suitable choice of the kernel and have found various applications in data-driven control, but only few approaches provide theoretical guarantees~\cite{fiedler2021practical}.

\newpage
Data-driven control approaches for LPV systems as in~\cite{formentin2016direct} can exploit prior knowledge by selecting a suitable controller parametrization.
Further, the development of loop-shaping controllers based on measured data was tackled in~\cite{karimi2017data}, in which case the choice of the filter represents prior model knowledge for robust controller design.

\vskip4pt
{
\noindent
\textit{Contribution}
}

\vskip2pt
In almost all practical applications, some prior knowledge on the underlying system is available, e.g., in the form of parameter values, plant structure, or, in a robust control context, transfer functions of weights used for loop-shaping design.
While the recent literature has seen a surge of contributions on data-driven control inspired by~\cite{willems2005note}, none of these methods can handle such prior knowledge in a flexible and general fashion.
However, ignoring available knowledge about the system is unduly restrictive when trying to design high-performance robust controllers.
In particular, if the data are not of high quality (i.e., persistently exciting~\cite{willems2005note} with a small noise level), then black-box approaches to data-driven control inherently produce unnecessarily conservative results and may lead to poor performance.

In this paper, we present a flexible framework for systematically combining measured data of an LTI system with prior knowledge for robust controller design.
The proposed framework relies on three main building blocks:
i) Modeling of the partially known system as a linear fractional transformation (LFT), which can elegantly separate known and unknown components;
ii) prior knowledge in the form of multipliers for the unknown components, which can describe bounds as well as structure; and
iii) an input-state trajectory which is affected by a disturbance admitting a general multiplier description and which
is exploited to learn additional multipliers for the unknown system components.
After combining these ingredients into a generic LFT formulation, we design robust controllers with closed-loop stability and performance guarantees for all uncertainties consistent with the prior knowledge and the data.
Due to its flexibility, our approach allows for seamless extensions into different directions, such as nonlinear uncertainties or output-feedback design from input-output data.


\vskip4pt
{
\noindent
\textit{Outline}
}

\vskip2pt
After introducing the problem setup in Section~\ref{sec:setting}, the multipliers inferred from prior knowledge and learnt from data are derived in Section~\ref{sec:parametrization}.
We employ these multipliers to design robust controllers with $\mathcal{H}_2$-performance guarantees in Section~\ref{sec:design}.
In Section~\ref{sec:prior}, we illustrate the flexibility of the proposed framework by discussing numerous forms of prior knowledge that can be covered.
We also address the conservatism of our design approach (Section~\ref{sec:necessity}) as well as controller design in the presence of additional nonlinear uncertainties (Section~\ref{sec:nonlinear}).
Numerical examples to demonstrate the advantages of our framework are presented in Section~\ref{sec:example}.
Finally, the paper is concluded in Section~\ref{sec:conclusion}.

\vskip4pt
{
\noindent
\textit{Notation}
}

\vskip2pt
We write $I_n$ for an $n\times n$ identity matrix, where the index is omitted if the dimension is clear from the context.
Further, we write $\mathrm{diag}_{i=1}^n(A_i)$ for the diagonal matrix with blocks $A_1,\dots,A_n$ and $A\otimes B$ for the Kronecker product of $A$ and $B$.
The space of square-summable sequences is denoted by $\ell_2$.
We write $\lVert x\rVert_{p}$ for the $p$-norm of a vector $x$.
In a matrix inequality, $\star$ represents blocks, which can be inferred from symmetry.
For some generic sequence $\{x_k\}_{k=0}^N$ we define
\begin{align*}
X&\coloneqq\begin{bmatrix}x_0&x_1&\dots&x_{N-1}\end{bmatrix},\>\>
X_+\coloneqq\begin{bmatrix}x_1&x_2&\dots&x_N\end{bmatrix}.
\end{align*}

\section{Problem setup}\label{sec:setting}
\begin{subequations}\label{eq:sys}
In this paper, we consider uncertain LTI systems of the form
\begin{align}\label{eq:sys1}
\left[
\begin{array}{c}x_{k+1}\\\hline e_k\\z_k\end{array}
\right]
&=\left[
\begin{array}{c|ccc}
A&B&B_d&B_w\\\hline
C_e&D_{eu}&D_{ed}&0\\
C_z&D_{z}&0&0\end{array}\right]
\left[
\begin{array}{c}
x_k\\\hline u_k\\d_k\\w_k
\end{array}\right]\\\label{eq:sys2}
w_k&=\Delta_{\tr} z_k,
\end{align}
\end{subequations}
where $x_k\in\mathbb{R}^n$ is the state vector, $u_k\in\mathbb{R}^m$ is the control input, $d_k\in\mathbb{R}^{n_d}$ is an external disturbance input, and $e_k\in\mathbb{R}^{n_e}$ is the performance output, all at time $k\geq0$.
In addition, the variables $w_k$ and $z_k$ represent an uncertainty channel and are of dimension $n_w$ and $n_z$, respectively.
We assume that all matrices in~\eqref{eq:sys} except for the \emph{true} uncertainty matrix $\Delta_{\tr}$ are known.
Hence,~\eqref{eq:sys} is an LFT consisting of an LTI system $\Sigma$ interconnected with a real-valued uncertainty $\Delta_{\tr}\in\mathbb{R}^{n_w\times n_z}$, compare Figure~\ref{fig:LFT}.
\begin{figure}
\begin{center}
\includegraphics[width=0.25\textwidth]{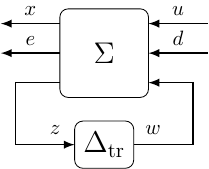}
\end{center}
\caption{Uncertain plant generating data for robust control.}
\label{fig:LFT}
\end{figure}
LFTs provide a flexible description of interconnections of known and unknown system components, and a wide variety of systems with uncertain parameters can be brought to the form~\eqref{eq:sys} (compare~\cite[Section 10]{zhou1996robust}).
In addition to the standard robust control interpretation, the LFT~\eqref{eq:sys} can also be interpreted as a partially known model where the known matrices in~\eqref{eq:sys} encode prior structural knowledge about the system while the uncertainty captures unknown parameters.
In this paper, we propose the use of LFTs for solving problems at the intersection of learning and robust control, where the inclusion of different forms of prior knowledge is of paramount importance.
We note that LFTs have also been used to include prior knowledge in a learning context in the recent work~\cite{holicki2020controller}.

\begin{example}\label{ex:introduction}
Let us illustrate the contributions of the paper by means of a realistic example from~\cite{franklin2019feedback},
a flexible satellite with pointing angle $\theta_1$ carrying an instrument package at angle $\theta_2$, cf.\ Figure~\ref{fig:satellite}.
\begin{figure}
\begin{center}
\includegraphics[width=0.4\textwidth]{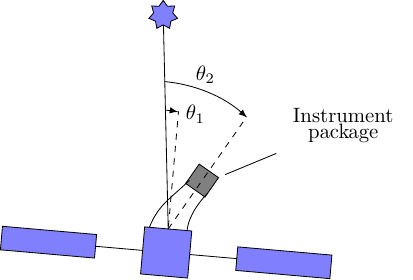}
\end{center}
\caption{Illustration of the satellite system in Example~\ref{ex:introduction}.
This figure as well as the example are adapted from~\cite{franklin2019feedback}.}
\label{fig:satellite}
\end{figure}
The dynamics of this system can be described by
%
\begin{align}\label{eq:ex_satellite}
\begin{bmatrix}\ddot{\theta}_2\\\ddot{\theta}_1\end{bmatrix}
=\begin{bmatrix}
-\frac{k}{J_2}\theta_2-\frac{b}{J_2}\dot{\theta}_2+\frac{k}{J_2}\theta_1+\frac{b}{J_2}\dot{\theta}_1+\frac{1}{J_2}d\\
\frac{k}{J_1}\theta_2+\frac{b}{J_1}\dot{\theta}_2-\frac{k}{J_1}\theta_1-\frac{b}{J_1}\dot{\theta}_1+\frac{1}{J_1}u
\end{bmatrix}
\end{align}
with (unknown) parameters $J_1=1$, $J_2=0.1$, $k=0.91$, $b=0.0036$.
The control input $u$ is a torque actuating the satellite and $d$ is a torque disturbance.
This system can be expressed by the LFT
\begin{align}\label{eq:ex_satellite_LFT}
\dot{x}&=\begin{bmatrix}0&1&0&0\\0&0&0&0\\0&0&0&1\\0&0&0&0\end{bmatrix}x
+\begin{bmatrix}0&0\\1&0\\0&0\\0&1\end{bmatrix}w+\begin{bmatrix}0\\1\\0\\0\end{bmatrix}\tilde{d},\\\nonumber
z&=\begin{bmatrix}I\\0\end{bmatrix}x+\begin{bmatrix}0\\1\end{bmatrix}u,\>\>
w=\Delta_{\tr}z
\end{align}
with $x\coloneqq\begin{bmatrix}\theta_2^\top&\dot{\theta}_2^\top&\theta_1^\top&\dot{\theta}_1^\top\end{bmatrix}^\top$, $\tilde{d}\coloneqq\frac{d}{J_2}$, and the uncertainty
\begin{align*}
\Delta_{\tr}\coloneqq\begin{bmatrix}-\frac{k}{J_2}&-\frac{b}{J_2}&\frac{k}{J_2}&\frac{b}{J_2}&0\\
\frac{k}{J_1}&\frac{b}{J_1}&-\frac{k}{J_1}&-\frac{b}{J_1}&\frac{1}{J_1}\end{bmatrix}.
\end{align*}
Suppose we have a set of (sampled) input-state measurements of~\eqref{eq:ex_satellite_LFT}, which are affected by some unknown but bounded disturbance $\tilde{d}$.
It is obviously possible to design a controller using these data via existing methods such as~\cite{persis2020formulas,berberich2020design,waarde2022from}.
However, it is clearly important to respect the structure of the description
\eqref{eq:ex_satellite_LFT}, which boils down to assuming
that the LFT~\eqref{eq:ex_satellite_LFT} is known except for the parameters collected in $\Delta_{\tr}$.
Moreover, from physical considerations, one might typically have access to (at least crude) bounds on $\Delta_{\tr}$.
Data-driven control methods as in~\cite{persis2020formulas,berberich2020design,waarde2022from}
can neither take
such bounds nor the structure of~\eqref{eq:ex_satellite_LFT} into account at all.
Therefore, unless the data are ideal (i.e., the input is persistently exciting and the noise level is zero), they could fail to achieve high performance for the closed loop.
The methods developed in this paper do allow to systematically include all this available information in the controller
design, 
thereby substantially reducing conservatism and improving performance if compared to purely data-driven approaches (see the numerical example in Section~\ref{sec:example}).
Further, as we show in Section~\ref{subsec:ex_Hinf},
the possibility to incorporate prior knowledge in the form of user-specified performance weights for $\mathcal{H}_{\infty}$-loop-shaping
significantly enhances the practical applicability of data-driven $\mathcal{H}_{\infty}$-control approaches
over existing ones.
\end{example}

Let us now describe the problem setting more precisely.
We assume that some prior information on the uncertainty $\Delta_{\tr}$ is available, i.e., $\Delta_{\tr}\in\bm{\Delta}_{\mathrm{prior}}$ for some known set $\bm{\Delta}_{\mathrm{prior}}\subseteq\mathbb{R}^{n_w\times n_z}$.
Moreover, we are given an input-state trajectory $\{x_k\}_{k=0}^N$, $\{u_k\}_{k=0}^{N-1}$ of~\eqref{eq:sys} for an \emph{unknown} disturbance sequence $\{\hat{d}_k\}_{k=0}^{N-1}$, which is collected in the matrix $D_{\mathrm{tr}}=\begin{bmatrix}\hat{d}_0&\dots&\hat{d}_{N-1}\end{bmatrix}$.
For the disturbance, we assume to have prior knowledge in the form of $D_{\tr}\in\bm{D}$ for some set $\bm{D}$
available.
With $X_+$, $X$ and $U$ defined for the given sequences $\{x_k\}_{k=0}^{N}$, $\{u_k\}_{k=0}^{N-1}$, let us note that
\eqref{eq:sys} implies $X_+=AX+BU+B_dD+B_w\Delta_{\tr}(C_zX+D_zU)$. With
\begin{align}\label{eq:MZ_def}
M&\coloneqq X_+-AX-BU,\>\>Z\coloneqq C_zX+D_{z}U,
\end{align}
this motivates to define the
set of \emph{learnt} uncertainties $\Delta$ as
\begin{align}\label{eq:Sigma_Delta_def}
\Del\coloneqq\Big\{\Delta\mid &M=B_w\Delta Z+B_dD\>\>\text{for some}\>\>D\in\bm{D}\Big\}.
\end{align}
In fact, this set comprises all uncertainties $\Delta$ that are consistent with the data and some disturbance in $\bm{D}$.
%
Note that the \emph{true} uncertainty is guaranteed to lie in $\Del$, i.e., $\Delta_{\tr}\in\Del$.
Therefore, it is also contained in the \emph{combined} uncertainty set
\begin{align}\label{eq:Decom_def}
&\Decom\coloneqq\Dep\cap\Del\\\nonumber
&=\left\{\Delta\in\Dep \mid M=B_w\Delta Z+B_dD\>\>\text{for some}\>\>D\in\bm{D}\right\}.
\end{align}
In this paper, we present a novel framework for combining prior knowledge with measured data via a tractable and tight description of the set $\Decom$ in terms of so-called multipliers (Section~\ref{sec:parametrization}).
Based on this description, we design state-feedback controllers $u_k=Kx_k$ with robust $\mathcal{H}_2$-performance guarantees (Section~\ref{subsec:design_H2}) and with quadratic performance guarantees in the presence of \emph{nonlinear} uncertainties (Section~\ref{sec:nonlinear}), as well as output-feedback controllers (Section~\ref{subsec:design_output}).

It is important to stress that prior knowledge about the uncertainty structure or bounds thereof, or of the system matrices in~\eqref{eq:sys}, is not required to apply the methodology in this paper. In fact,
our results contain a data-driven setting without any prior knowledge at all as a special case, cf.\ Example~\ref{ex:data1}. 
Nevertheless, any additional prior knowledge may shrink the set $\Decom$ and can thus reduce conservatism and improve the performance of the resulting robust controller if compared to a design based only on the available data.
On the other hand, our approach also applies if no data are available but only prior model knowledge and an uncertainty description are given, i.e., $\Decom=\Dep$.
In this sense,
the paper resolves the dichotomy of purely model-based and purely data-driven control.
Not only do we arrive at a language
for modeling problems that are positioned between these two established scenarios, but
we also showcase its flexibility in that it enables a variety of unprecedented extensions.

\subsubsection*{Prior knowledge on the uncertainty}
Let us describe the prior knowledge $\Delta_{\tr}\in\bm{\Delta}_{\mathrm{prior}}$ in more detail.
We consider a block-diagonal structure, i.e.,
\begin{align}\label{eq:prior_structure}
\Delta_{\tr}=\text{diag}_{j=1}^{\ell}(\Delta_j),
\end{align}
where $\Delta_j\in\mathbb{R}^{n_{w,j}\times n_{z,j}}$ is either a full matrix block or a repeated scalar block $\Delta_j=\delta_jI_{n_{w,j}}$, $\delta_j\in\mathbb{R}$, which can capture multiple occurrences of the same uncertain parameter in~\eqref{eq:sys}.
For $\Delta\in\Dep$, we make use of the description
\begin{align}\label{eq:Delta_diag_L_R}
    \Delta &=\mathrm{diag}_{j=1}^{\ell}(\Delta_j)=\sum_{j=1}^{\ell}R_j\Delta_j L_j^T,
\end{align}
where $L_j$ and $R_j$ are the corresponding  block-columns of the identity matrix (satisfying $L_k^\top L_j=0$, $R_k^\top R_j=0$ for $k\neq j$).
We partition $B_w$ according to the structure of $\Delta$ as $B_w=\begin{bmatrix}B_1&\dots&B_{\ell}\end{bmatrix}$ with $B_j\coloneqq B_w R_j$, and we assume without loss of generality that the matrices $B_j$ have full column rank.
Then,~\eqref{eq:Delta_diag_L_R} implies $\Delta L_j=R_j\Delta_j$ and thus
\begin{align}\label{eq:Bw_Delta_Lj_equation}
    B_w\Delta L_j=B_j\Delta_j.
\end{align}
We assume that a multiplier description for the individual uncertainties $\Delta_j$ is available as follows.
\begin{assumption}\label{ass:prior_multipliers}
The set $\bm{\Delta}_{\mathrm{prior}}$ describing the prior knowledge on $\Delta_{\tr}\in\Dep$ is described by
\begin{align}\label{eq:prior_set_Delta}
\bm{\Delta}_{\mathrm{prior}}=\left\{\Delta=\mathrm{diag}_{j=1}^{\ell}(\Delta_j)\mid \Delta_j\in\bm{\Delta}_j\right\},
\end{align}
where $\bm{\Delta}_j$ is defined as
\begin{align}\label{eq:prior_set_Delta_i}
    \bm{\Delta}_j=\left\{\Delta_j\in\mathbb{R}^{n_{w,j}\times n_{z,j}}\Bigm|
    \begin{bmatrix}\Delta_j^\top\\I\end{bmatrix}^\top\!\!
    P_j\begin{bmatrix}\Delta_j^\top\\I\end{bmatrix}\succeq0\>\>\forall P_j\in\bm{P}_j\right\}
\end{align}
and where $\bm{P}_j$ are convex cones of symmetric matrices admitting an LMI representation. This means that they can be expressed as the projection of the solution set of a strictly feasible LMI.
Moreover, for any $j=1,\dots,\ell$, there exists some $P_j\in\bm{P}_j$ such that $\begin{bmatrix}I_{n_{z,j}}&0\end{bmatrix} P_j\begin{bmatrix}I_{n_{z,j}}&0\end{bmatrix}^\top\prec0$ holds true.
\end{assumption}
As discussed with various examples in Section~\ref{subsec:prior_uncertainty}, Assumption~\ref{ass:prior_multipliers} provides a quite general and flexible description of possible prior knowledge on the uncertainty $\Delta_{\mathrm{tr}}$.

\subsubsection*{Prior knowledge on the disturbance}
In addition to the above prior uncertainty description $\Dep$, an input-state trajectory $\{x_k\}_{k=0}^N$, $\{u_k\}_{k=0}^{N-1}$ of~\eqref{eq:sys} generated by an unknown disturbance sequence $\{\hat{d}_k\}_{k=0}^{N-1}$ is available.
Throughout this paper, we assume that the corresponding input matrix $B_d$ has full column rank.
This causes no loss of generality since, otherwise, we can define a different disturbance $\tilde{d}$ which exerts the same influence on~\eqref{eq:sys} while $B_{\tilde{d}}$ has full column rank.
Further, we assume that $D_{\mathrm{tr}}=\begin{bmatrix}\hat{d}_0&\dots&\hat{d}_{N-1}\end{bmatrix}$ satisfies a multiplier description analogous to Assumption~\ref{ass:prior_multipliers}.
\begin{assumption}\label{ass:noise_multipliers}
The noise which generates the available sequence satisfies $D_{\tr}\in\bm{D}$, where
\begin{align}\label{eq:noise_multipliers}
\bm{D}\coloneqq\Big\{D\in\mathbb{R}^{n_d\times N}\Bigm|\begin{bmatrix}D^\top\\I\end{bmatrix}^\top\!\! P_d\begin{bmatrix}D^\top\\I\end{bmatrix}\succeq0
\>\>\forall P_d\in\bm{P}_d\Big\}
\end{align}
and where $\bm{P}_d$ is a convex cone of symmetric matrices admitting an LMI representation.
Moreover, there exists $P_d\in\bm{P}_d$ such that $\begin{bmatrix}I_N&0\end{bmatrix}P_d\begin{bmatrix}I_N&0\end{bmatrix}^\top\prec0$.
\end{assumption}
In Section~\ref{subsec:prior_disturbance}, we illustrate the flexibility of Assumption~\ref{ass:noise_multipliers} by discussing various important special cases.

\begin{example}\label{ex:data1}
Let us consider the purely data-driven case, i.e.,
\begin{align}\label{eq:ex_data1}
x_{k+1}=A_{\tr}x_k+B_{\tr}u_k+B_dd_k,
\end{align}
where $A_{\tr}$, $B_{\tr}$ are unknown and $B_d$ is known.
This system can be brought to the form~\eqref{eq:sys} by choosing $A=0$, $B=0$, $B_w=I$, $C_z=\begin{bmatrix}I&0\end{bmatrix}^\top$, $D_{z}=\begin{bmatrix}0&I\end{bmatrix}^\top$, and $\Delta_{\tr}=\begin{bmatrix}A_{\tr}&B_{\tr}\end{bmatrix}$.
The recent papers~\cite{persis2020formulas,berberich2020design,waarde2022from} have addressed data-driven controller design for~\eqref{eq:ex_data1} with a specific choice of $\bm{P}_d$ and without additional prior knowledge on $A_{\tr}$ or $B_{\tr}$ (i.e., $\Dep=\mathbb{R}^{n\times(n+m)}$).
As we show and discuss throughout the manuscript, the flexibility gained via the proposed framework, in particular by choosing more general multipliers $\bm{P}_d$ and prior descriptions $\Dep$, allows for a significant generalization and improvement in comparison to these earlier works.
\end{example}

\section{Multipliers from prior knowledge and data}\label{sec:parametrization}

In this section, we present a systematic approach to construct multipliers for $\Delta_{\tr}$ based on prior knowledge and noisy data.
First, we show how prior multipliers as in Assumption~\ref{ass:prior_multipliers} can be equivalently translated into multipliers for a transformed full-block uncertainty in Section~\ref{subsec:param_prior}.
Next, in Section~\ref{subsec:param_data}, we use data affected by a disturbance satisfying Assumption~\ref{ass:noise_multipliers} in order to learn additional multipliers for this transformed uncertainty.
Combining these results, we obtain multipliers for all uncertainties $\Delta\in\Decom$ consistent with the prior knowledge and the data (Section~\ref{subsec:param_combined}).

\subsection{Multipliers from prior knowledge}\label{subsec:param_prior}
Let us consider the uncertain system
\begin{align}\label{eq:sys_prior}%
\text{\eqref{eq:sys1} with}\>\>w_k=\Delta z_k
\end{align}%
for $\Delta\in\bm{\Delta}_{\mathrm{prior}}$ with
$\Dep$ given as in Assumption~\ref{ass:prior_multipliers}.
Since $\Delta_{\tr}\in\bm{\Delta}_{\mathrm{prior}}$, this family of systems comprises the true system~\eqref{eq:sys} under consideration.
Note that $\tilde{w}_k\coloneqq B_ww_k$ satisfies $\tilde{w}_k=B_w\Delta z_k=\tilde{\Delta}z_k$ with $\tilde{\Delta}\coloneqq B_w\Delta$.
Thus,~\eqref{eq:sys_prior} is equivalent to
\begin{subequations}\label{eq:sys_prior_tilde}%
\begin{align}\label{eq:sys_prior_tilde1}
\left[
\begin{array}{c}x_{k+1}\\\hline e_k\\z_k\end{array}
\right]
&=\left[
\begin{array}{c|ccc}
A&B&B_d&I\\\hline
C_e&D_{eu}&D_{ed}&0\\
C_z&D_{z}&0&0\end{array}\right]
\left[
\begin{array}{c}
x_k\\\hline u_k\\d_k\\\tilde{w}_k
\end{array}\right]\\\label{eq:sys_prior_tilde2}
\tilde{w}_k&=\tilde{\Delta} z_k,
\end{align}
\end{subequations}
where $\tilde{\Delta}\in B_w\bm{\Delta}_{\mathrm{prior}}$.
We partition the columns of $\tilde{\Delta}=\begin{bmatrix}\tilde{\Delta}_1&\dots&\tilde{\Delta}_{\ell}\end{bmatrix}$ according to the structure of $B_w\Delta$, i.e., $\tilde{\Delta}_j\coloneqq \tilde{\Delta}L_j=B_w\Delta L_j\stackrel{\eqref{eq:Bw_Delta_Lj_equation}}{=}B_j\Delta_j$.
Then, we define the convex cone of transformed and combined multipliers
\begin{align}\label{eq:prior_multiplier_transformed}
    \bm{\tilde{P}}\coloneqq \sum_{j=1}^{\ell}\begin{bmatrix}L_j^\top&0\\0&B_j^\top\end{bmatrix}^\top\bm{P}_j
    \begin{bmatrix}L_j^\top&0\\0&B_j^\top\end{bmatrix}
\end{align}
as well as the corresponding set of uncertainties
\begin{align}\label{eq:tilde_Delta_prior_def}
    \tilde{\bm{\Delta}}_{\mathrm{prior}}\coloneqq\left\{
    \tilde{\Delta}\in\mathbb{R}^{n\times n_z}\Bigm|\begin{bmatrix}\tilde{\Delta}^\top\\I\end{bmatrix}^\top\!\!\!\tilde{P}
    \begin{bmatrix}\tilde{\Delta}^\top\\I\end{bmatrix}\succeq0\>\>\forall\tilde{P}\in\bm{\tilde{P}}\right\}.
\end{align}
The following result shows that the uncertainty description $\tilde{\Delta}\in\tilde{\bm{\Delta}}_{\mathrm{prior}}$ is \emph{equivalent} to $\tilde{\Delta}\in B_w\bm{\Delta}_{\mathrm{prior}}$.
\begin{lemma}\label{lem:prior_bounds}
If Assumption~\ref{ass:prior_multipliers} holds, then $\tilde{\bm{\Delta}}_{\mathrm{prior}}=B_w\bm{\Delta}_{\mathrm{prior}}$.
\end{lemma}

The proof of Lemma~\ref{lem:prior_bounds} is provided in Appendix~\ref{app:A}.
Lemma~\ref{lem:prior_bounds} shows that $\bm{\tilde{P}}$ in~\eqref{eq:prior_multiplier_transformed} is a valid class of multipliers for the \emph{transformed} uncertainty $\tilde{\Delta}$.
Notably, $\bm{\tilde{P}}$ takes the prior bounds on $\Delta$ in Assumption~\ref{ass:prior_multipliers} into account without losing the information that $\tilde{\Delta}$ is of the form $\tilde{\Delta}=B_w\Delta$.
This is made possible by allowing for \emph{non-strict} matrix inequalities in~\eqref{eq:tilde_Delta_prior_def}.
To the best of our knowledge, the insight of translating structural knowledge for the uncertainty channel into multipliers for a transformed full-block uncertainty is new.
It is a key step to consider $\tilde{\Delta}$ instead of $\Delta$, since it is possible to systematically learn multipliers from data for $\tilde{\Delta}$, as seen next.

\subsection{Learning multipliers from data}\label{subsec:param_data}
Recall that we have available an input-state trajectory of~\eqref{eq:sys} for some unknown disturbance sequence satisfying Assumption~\ref{ass:noise_multipliers}. These data permit to learn the new constraint $\Del$ in \eqref{eq:Sigma_Delta_def}
on the uncertainties $\Delta$, with the guarantee that it still contains the true uncertainty, i.e., $\Det\in\Del$.
In the following, we express the learnt information by a family of quadratic constraints that is amenable to robust control.
For this purpose, we need to work with a parametrization of $B_w\Del$.
With $M$ and $Z$ in~\eqref{eq:MZ_def} for the corresponding matrices depending on the data and $A$, $B$, $C_z$, $D_{z}$, we introduce the multiplier class
\begin{align}\label{eq:data_multipliers}
    \bm{\tilde{P}}_0&\coloneqq\begin{bmatrix}-Z^\top&M^\top\\0&B_d^\top\end{bmatrix}^\top
    \bm{P}_d
    \begin{bmatrix}-Z^\top&M^\top\\0&B_d^\top\end{bmatrix}
\end{align}
as well as the corresponding uncertainty set
\begin{align*}
    \tDel\coloneqq\Big\{\tilde{\Delta}\Bigm|
    \begin{bmatrix}\tilde{\Delta}^\top\\I\end{bmatrix}^\top
    \tilde{P}_0\begin{bmatrix}\tilde{\Delta}^\top\\I\end{bmatrix}\succeq0\>\>\text{for all}\>\>\tilde{P}_0\in\bm{\tilde{P}}_0\Big\}.
\end{align*}
\begin{lemma}\label{lem:data_bound}
If Assumption~\ref{ass:noise_multipliers} holds, then $B_w\Del\subseteq\tDel$.
\end{lemma}
The proof of Lemma~\ref{lem:data_bound} is provided in Appendix~\ref{app:B}.
Lemma~\ref{lem:data_bound} states that $\bm{\tilde{P}}_0$ is a valid class of multipliers for all transformed uncertainties consistent with the data.
Since $\Det\in\Del$, this also implies $B_w\Det\in\tDel$, i.e., we obtain a constraint on the unknown true uncertainty $\Det$ without explicitly identifying the underlying system.
Computing the multipliers $\bm{\tilde{P}}_0$ requires knowledge of the data matrices $X$, $X_+$ and $U$ 
the prior model parts $A$, $B$, $B_d$, $C_z$, $D_{z}$, as well as the disturbance description $\bm{D}$.
Note that the set $\tDel$ only provides an \emph{exact} parametrization of all $\tilde{\Delta}$ satisfying
\begin{align*}
    M=B_dD+\tilde{\Delta}Z\>\>\text{holds for some}\>\>D\in\bm{D},
\end{align*}
but not for the set $B_w\Del$ (since $B_w\Del\supseteq\tDel$ does not hold in general).
This is due to the fact that $\tDel$ itself does not take into account the structural information that $\tilde{\Delta}$ is of the form $B_w\Delta$, which we include in our framework via the prior multipliers derived in Lemma~\ref{lem:prior_bounds}.

\subsection{Combined multipliers}\label{subsec:param_combined}

Recall that, by definition, $\tilde{\Delta}\in\tDep$ and $\tilde{\Delta}\in\tDel$ iff
\begin{align}\label{eq:combined_multipliers_prior_and_data}
    \begin{bmatrix}\tilde{\Delta}^\top\\I\end{bmatrix}^\top\tilde{P}
    \begin{bmatrix}\tilde{\Delta}^\top\\I\end{bmatrix}\succeq0\>\>\text{and}\>\>
    \begin{bmatrix}\tilde{\Delta}^\top\\I\end{bmatrix}^\top\tilde{P}_0
    \begin{bmatrix}\tilde{\Delta}^\top\\I\end{bmatrix}\succeq0
\end{align}
for all $\tilde{P}\in\bm{\tilde{P}}$, $\tilde{P}_0\in\bm{\tilde{P}}_0$, where $\bm{\tilde{P}}$ and $\bm{\tilde{P}}_0$ are the multiplier classes resulting from prior knowledge and data, respectively.
Hence, for any $\tilde{\Delta}\in\tDep\cap\tDel$, we infer
\begin{align}\label{eq:combined_multipliers_prior_and_data2}
        \begin{bmatrix}\tilde{\Delta}^\top\\I\end{bmatrix}^\top\!(\tilde{P}+\tilde{P}_0)
    \begin{bmatrix}\tilde{\Delta}^\top\\I\end{bmatrix}\succeq0\>\>
    \text{for all}\>\>
    \tilde{P}\in\bm{\tilde{P}},\>\tilde{P}_0\in\bm{\tilde{P}}_0.
\end{align}
Conversely, since $0\!\in\!\bm{\tilde{P}}$, $0\!\in\!\bm{\tilde{P}}_0$, any $\tilde{\Delta}$ with~\eqref{eq:combined_multipliers_prior_and_data2} also satisfies \eqref{eq:combined_multipliers_prior_and_data} for all $\tilde{P}\in\bm{\tilde{P}}$, $\tilde{P}_0\in\bm{\tilde{P}}_0$.
For the combined multiplier class $\bm{\tilde{P}}_{\mathrm{com}}\coloneqq\bm{\tilde{P}}+\bm{\tilde{P}}_0$
and the transformed, combined uncertainty set $\tDecom\coloneqq\tDep\cap\tDel$, this shows
\begin{align}\label{eq:tDecom_def}
    \tDecom=\left\{\tilde{\Delta}\Bigm|
    \begin{bmatrix}\tilde{\Delta}^\top\\I\end{bmatrix}^\top\!\!\!
    \tilde{P}_{\mathrm{com}}
    \begin{bmatrix}\tilde{\Delta}^\top\\I\end{bmatrix}\succeq0\>\>\forall\tilde{P}_{\mathrm{com}}\in\bm{\tilde{P}}_{\mathrm{com}}\right\}.
\end{align}
%
As explained in Sections~\ref{subsec:param_prior} and~\ref{subsec:param_data}, the (transformed) true uncertainty $\tDet$ satisfies $\tDet\in\tDep$ as well as $\tDet\in\tDel$ such that $\tDet\in\tDecom$.
In fact, the set $\tDecom$ combines all available information on $\tDet$ in an exact fashion.

\begin{lemma}\label{lem:combined_multipliers}
Under Assumptions~\ref{ass:prior_multipliers} and~\ref{ass:noise_multipliers}, $\tDecom=B_w\Decom$.
\end{lemma}
The proof of Lemma~\ref{lem:combined_multipliers} is provided in Appendix~\ref{app:C}.
By Lemma~\ref{lem:combined_multipliers}, $\tDecom$ contains all available information on $\tDet$, including the fact that it is of the form $B_w\Delta$ for some $\Delta\in\Decom$.
This structure is inherited from the prior description $\tDep\supseteq\tDecom$ due to Lemma~\ref{lem:prior_bounds} and is not captured by the learnt multipliers $\bm{\tilde{P}}_0$ alone.
Therefore, using multipliers $\bm{P}_j$ for prior knowledge as in Assumption~\ref{ass:prior_multipliers} is \emph{always} beneficial in the presence of structure, even if these multipliers are derived from very conservative bounds on $\Delta_j$.

\begin{example}\label{ex:data2}
Let us illustrate our findings via the purely data-driven special case with $\Delta_{\tr}=\begin{bmatrix}A_{\tr}&B_{\tr}\end{bmatrix}$ and $\Dep=\mathbb{R}^{n\times(n+m)}$, cf.\ Example~\ref{ex:data1}.
Suppose the disturbance generating the data satisfies $D_{\tr}D_{\tr}^\top\preceq \bar{d}I$ for some $\bar{d}>0$, i.e., Assumption~\ref{ass:noise_multipliers} holds with $\bm{P}_d=\{\lambda\cdot\mathrm{diag}(-I,\bar{d}I)\mid\lambda>0\}$.
Lemma~\ref{lem:prior_bounds} is trivial.
Further, Lemma~\ref{lem:data_bound} implies that any $\Delta=\begin{bmatrix}A&B\end{bmatrix}\in\Del$ consistent with the data satisfies
\begin{align*}
(X_+-AX-BU)(X_+-AX-BU)^\top\preceq \bar{d}B_dB_d^\top.
\end{align*}
The latter inequality means that $\tDel$, which equals $\Del=\tDecom=\Decom$ due to $\Dep=\mathbb{R}^{n\times(n+m)}$ and $B_w=I$, contains all matrices $\begin{bmatrix}A&B\end{bmatrix}$ for which the violation of the nominal system dynamics is bounded by the noise level.
For this special case,~\cite{waarde2022from} uses a parametrization analogous to Lemma~\ref{lem:data_bound} for data-driven control with stability and $\mathcal{H}_2$- or $\mathcal{H}_{\infty}$-performance objectives.
As a key advantage of this approach in comparison to earlier works such as~\cite{persis2020formulas,persis2021low,berberich2020design,waarde2020informativity}, the resulting design is less conservative and computationally more efficient since the number of decision variables is independent of the data length.
In case that no prior knowledge is available and if we use simple disturbance multipliers as above, Lemma~\ref{lem:data_bound} as well as the results in Section~\ref{subsec:design_H2} reduce to~\cite{waarde2022from}.
However, our framework provides the following generalizations and improvements over~\cite{waarde2022from}:
1) The possibility to include prior knowledge in terms of uncertainty structure and/or bounds (Section~\ref{subsec:prior_uncertainty});
2) the use of more flexible disturbance multipliers (Section~\ref{subsec:prior_disturbance});
3) the design of controllers based on input-output data (Section~\ref{subsec:design_output});
4) and the inclusion of nonlinear uncertainties (Section~\ref{sec:nonlinear}).
\end{example}



\section{Robust controller design using prior knowledge and data}\label{sec:design}

In this section, we use the representation of $\tDecom$ in~\eqref{eq:tDecom_def} to design controllers with guaranteed robust $\mathcal{H}_2$-performance of the channel $d\mapsto e$ for all $\Delta\in\Decom$ (Section~\ref{subsec:design_H2}).
Further, we show how our approach can be applied for output-feedback controller design based on input-output data in Section~\ref{subsec:design_output}.

\subsection{Robust $\mathcal{H}_2$-performance}\label{subsec:design_H2}

The $\mathcal{H}_2$-norm of an LTI system can be defined based on its frequency response, see~\cite[Section 3.3.3]{scherer2000linear}.
As a deterministic interpretation, the squared $\mathcal{H}_2$-norm is equal to the sum of the output energies of the system responses $e$ when applying impulsive inputs to the system in each component of $d$.
Furthermore, when choosing $C_e=\begin{bmatrix}\bar{C}^\top&0\end{bmatrix}^\top$, $D_{eu}=\begin{bmatrix}0&\bar{D}^\top\end{bmatrix}^\top$ for some $\bar{C}$, $\bar{D}$, minimizing the $\mathcal{H}_2$-norm of~\eqref{eq:sys} is equivalent to a linear-quadratic regulation problem with weighting matrices $Q=\bar{C}^\top\bar{C}$ and $R=\bar{D}^\top\bar{D}$ for the state and input, respectively.
The following result provides a design procedure in order to guarantee a closed-loop $\mathcal{H}_2$-performance bound for the channel $d\mapsto e$ based on the uncertainty parametrization~\eqref{eq:tDecom_def}.

\begin{theorem}\label{thm:robust_H2}
Suppose Assumptions~\ref{ass:prior_multipliers} and~\ref{ass:noise_multipliers} hold, $D_{ed}=0$, and there exist ${\mathcal{X}\succ0}$, $K$, $\tilde{P}_{\mathrm{com}}\in\bm{\tilde{P}}_{\mathrm{com}}$, $\gamma>0$
such that
\begin{align}\label{eq:prop_rob_H2_LMI1}
&\tr\big((C_e+D_{eu}K)\mathcal{X}(C_e+D_{eu}K)^\top\big)<\gamma^2
\end{align}
and~\eqref{eq:prop_rob_H2_LMI2} hold.
Then, for any $\Delta\in\Decom$, the system~\eqref{eq:sys_prior} controlled as $u_k=Kx_k$ is stable and $d\mapsto e$ has a closed-loop $\mathcal{H}_2$-norm less than $\gamma$.

\begin{figure*}
\vspace{0pt}
\begin{align}\label{eq:prop_rob_H2_LMI2}
\begin{bmatrix}I&0\\(A+BK)^\top&(C_z+D_zK)^\top\\\hline
0&I\\I&0\end{bmatrix}^\top
\left[
\begin{array}{c|c}
\begin{matrix}B_dB_d^\top-\mathcal{X}&0\\0&\mathcal{X}\end{matrix}
&\begin{matrix}0&0\\0&0\end{matrix}\\\hline
\begin{matrix}\qquad0\qquad\,\,&0\\\qquad0\qquad\,\,&0\end{matrix}&
\tilde{P}_{\mathrm{com}}\end{array}
\right]
\begin{bmatrix}I&0\\(A+BK)^\top&(C_z+D_zK)^\top\\\hline
0&I\\I&0\end{bmatrix}\prec0
\end{align}
\noindent\makebox[\linewidth]{\rule{\textwidth}{0.4pt}}
\end{figure*}
\end{theorem}
\begin{proof}
It follows from~\cite[Theorem 10.2]{scherer2000robust} that~\eqref{eq:prop_rob_H2_LMI2} implies
\begin{align}\label{eq:prop_rob_H2_proof1}
\mathcal{A}(\tilde{\Delta})\mathcal{X}\mathcal{A}(\tilde{\Delta})^\top-\mathcal{X}+B_pB_p^\top&\prec0
\end{align}
for all $\tilde{\Delta}\in\tDecom$, where $\mathcal{A}(\tilde{\Delta})\coloneqq A+BK+\tilde{\Delta}(C_z+D_zK)$, compare~\eqref{eq:tDecom_def}.
Using~\cite[Proposition 3.13]{scherer2000linear},~\eqref{eq:prop_rob_H2_LMI1} and~\eqref{eq:prop_rob_H2_proof1} together with $D_{ed}=0$ imply that~\eqref{eq:sys_prior_tilde} is stable and has a closed-loop $\mathcal{H}_2$-norm less than $\gamma$ for all $\tilde{\Delta}\in\tDecom$.
Since $B_w\Decom\subseteq\tDecom$, the same stability and performance properties also hold for~\eqref{eq:sys_prior} and any $\Delta\in\Decom$, which concludes the proof.
\end{proof}
Recall that the multiplier class $\bm{\tilde{P}}_{\mathrm{com}}$ is taken as $\bm{\tilde{P}}_{\mathrm{com}}=\bm{\tilde{P}}+\bm{\tilde{P}}_0$ with $\bm{\tilde{P}}$ and $\bm{\tilde{P}}_0$ as in~\eqref{eq:prior_multiplier_transformed} and~\eqref{eq:data_multipliers}, respectively, compare Section~\ref{subsec:param_combined} for details.
If the LFT~\eqref{eq:sys_prior} is interpreted as a partially known system with an unknown parameter $\Delta_{\tr}$, compare~\eqref{eq:sys}, we note that any robust controller designed via Theorem~\ref{thm:robust_H2} leads to the same guarantees for the \emph{true} unknown system.
In the proof, we show robust performance w.r.t.\ the transformed uncertainty $\tilde{\Delta}=B_w\Delta$ since the multiplier description in~\eqref{eq:tDecom_def} is only valid for $\tilde{\Delta}$, but not for $\Delta$ directly.
However, treating $B_w\Delta$ as one full-block uncertainty does not cause additional conservatism in the design since, by Lemma~\ref{lem:combined_multipliers}, the set $\tDecom$ takes into account the fact that $\tilde{\Delta}$ has the form $B_w\Delta$ for some $\Delta\in\Decom$.

While the matrix inequality~\eqref{eq:prop_rob_H2_LMI2} is not linear in the state-feedback gain $K$, it is simple to transform~\eqref{eq:prop_rob_H2_LMI2} into an LMI following standard steps (compare, e.g.,~\cite[Section 4.5]{scherer2000linear}).
By just applying the Schur complement and defining the new variable $L=K\mathcal{X}$,~\eqref{eq:prop_rob_H2_LMI2} is equivalent to
\begin{align*}
\begin{bmatrix}\begin{bmatrix}B_dB_d^\top-\mathcal{X}&0\\0&0\end{bmatrix}+
\begin{bmatrix}0&I\\I&0\end{bmatrix}^\top \tilde{P}_{\mathrm{com}}\begin{bmatrix}0&I\\I&0\end{bmatrix}&
{\text{\large$\star$}}\\\\
\begin{bmatrix}\mathcal{X}A^\top+L^\top B^\top&\mathcal{X}C_z^\top+L^\top D_z^\top\end{bmatrix}&-\mathcal{X}
\end{bmatrix}\prec0,
\end{align*}
which is an LMI in the variables $\mathcal{X}$ and $L$.
Similarly,~\eqref{eq:prop_rob_H2_LMI1} is equivalent to the existence of $\Gamma$ such that $\mathrm{tr}(\Gamma)<\gamma^2$ and
\begin{align*}
    \begin{bmatrix}\Gamma&C_e\mathcal{X}+D_{eu}L\\(C_e\mathcal{X}+D_{eu}L)^\top&\mathcal{X}\end{bmatrix}\succ0.
\end{align*}
While the above result focuses on $\mathcal{H}_2$-performance, Theorem~\ref{thm:robust_H2} also allows us to design controllers which robustly stabilize~\eqref{eq:sys_prior} for all $\Delta\in\Decom$ by simply omitting the constraint~\eqref{eq:prop_rob_H2_LMI1} and setting $B_d=0$ in~\eqref{eq:prop_rob_H2_LMI2}.
Designing for robust quadratic performance is possible as well and will be addressed in a more general problem setting with nonlinear uncertainties in Section~\ref{sec:nonlinear}.
Finally, the presented results can be trivially extended to the case that performance is ensured for some channel $d'\mapsto e$ with a \emph{different} disturbance input $d'$, not necessarily equal to the disturbance $d$ generating the data.

Theorem~\ref{thm:robust_H2} has multiple advantages if compared to a sequential approach using system identification and model-based robust control.
The synthesis conditions rely on the learnt multipliers provided by Lemma~\ref{lem:data_bound} which are tight and computationally attractive.
In particular, they only use noisy data of finite length and do not require an additional estimation procedure.
On the other hand, providing similar error bounds using system identification is in general difficult and an active field of research~\cite{dean2019sample,matni2019self,matni2019tutorial}.
If compared to methods based on set membership estimation, see~\cite{belforte1990parameter,milanese1991optimal}, our approach is more flexible due to the generality of the multiplier classes in Assumptions~\ref{ass:prior_multipliers} and~\ref{ass:noise_multipliers} (compare Section~\ref{sec:prior}), and
more direct, combining all available knowledge in the controller design without any step of precomputation.

\subsection{Robust output-feedback design}\label{subsec:design_output}
In this section, we illustrate how our approach can be used to design robust output-feedback controllers based on noisy input-output data.
As we will see, data-driven output-feedback design leads to a structured robust control problem which fits naturally into our framework.
We consider systems of the form
\begin{align}\label{eq:sys_IO}
y_{k}=&A_1y_{k-1}+\dots+A_{n}y_{k-n}\\\nonumber
&+B_0u_k+\dots+B_{n}u_{k-n}+B_d^0d_k,
\end{align}
where $y_k\in\mathbb{R}^p$ is the output, $u_k\in\mathbb{R}^m$ is the input, $d_k\in\mathbb{R}^{n_d}$ is a disturbance, $n$ is the system order, and all matrices except for $B_d^0$ are unknown.
\begin{figure*}
\vspace{0pt}
\begin{align}\label{eq:sys_IO_state_space}
\left[\begin{array}{c}
u_{k-n+1}\\\vdots\\u_{k-1}\\u_k\\\hline y_{k-n+1}\\\vdots\\y_{k-1}\\y_k
\end{array}\right]
=
\left[\begin{array}{cccc|cccc}
0&I&\dots&0&0&\dots&\dots&0\\
\vdots&\ddots&\ddots&\vdots&\vdots&\ddots&\ddots&\vdots\\
0&\ddots&\ddots&I&\vdots&\ddots&\ddots&\vdots\\
0&\dots&\dots&0&0&\dots&\dots&0
\\\hline
0&\dots&\dots&0&0&I&\dots&0\\
\vdots&\ddots&\ddots&\vdots&\vdots&\ddots&\ddots&\vdots\\
0&\dots&\dots&0&0&\dots&\dots&I\\
B_n&\dots&\dots&B_1&
A_n&\dots&\dots&A_1
\end{array}\right]
\left[\begin{array}{c}
u_{k-n}\\\vdots\\u_{k-2}\\u_{k-1}\\\hline y_{k-n}\\\vdots\\y_{k-2}\\y_{k-1}
\end{array}\right]
+
\left[\begin{array}{c}
0\\\vdots\\0\\I\\\hline 0\\\vdots\\0\\B_0
\end{array}\right]
u_k
+
\left[\begin{array}{c}
0\\\vdots\\0\\0\\\hline 0\\\vdots\\0\\B_d^0
\end{array}\right]
d_k
\end{align}
\noindent\makebox[\linewidth]{\rule{\textwidth}{0.4pt}}
\end{figure*}
\begin{subequations}\label{eq:sys_IO_LFT}
It is straightforward to see (compare, e.g.,~\cite[Lemma 3.4.7]{goodwin2014adaptive}, \cite{koch2022provably,persis2020formulas}) that~\eqref{eq:sys_IO} can be written equivalently as the state-space system~\eqref{eq:sys_IO_state_space} with the extended state $\xi_k=\begin{bmatrix}u_{k-n}^\top&\dots&u_{k-1}^\top&y_{k-n}^\top&\dots&y_{k-1}^\top\end{bmatrix}^\top$.
Using that~\eqref{eq:sys_IO_state_space} contains both known components as well as unknown parameters in the last row, it can be written as the LFT
\begin{align}
\xi_{k+1}&=A\xi_k+Bu_k+B_ww_k+B_dd_k,\\
z_k&=\begin{bmatrix}I\\0\end{bmatrix}\xi_k+\begin{bmatrix}0\\I\end{bmatrix}u_k,\\
w_k&=\Delta_{\tr} z_k,
\end{align}
\end{subequations}
where $A$, $B$, $B_w$, $B_d$ are suitably defined \emph{known} matrices and $\Delta_{\tr}=\begin{bmatrix}B_n&\dots&B_1&A_n&\dots&A_1&B_0\end{bmatrix}$ plays the role of the uncertainty.
Clearly,~\eqref{eq:sys_IO_LFT} is of the form~\eqref{eq:sys} such that all results in this paper are applicable to the system~\eqref{eq:sys_IO_LFT}.
More precisely, suppose that measurements of the input $\{u_k\}_{k=0}^{N-1}$ and the extended state $\{\xi_k\}_{k=0}^{N}$ are available, corresponding to input-output measurements of~\eqref{eq:sys_IO}, and affected by an unknown noise sequence $\{\hat{d}_k\}_{k=0}^{N-1}$ satisfying a known description $D_{\mathrm{tr}}\in\bm{D}$ with $\bm{D}$ as in Assumption~\ref{ass:noise_multipliers}.
Then, the proposed framework can be utilized to design robust state-feedback controllers $u_k=K\xi_k$ for~\eqref{eq:sys_IO_LFT} based on the measured data.
Due to the above definition of $\xi_k$, the resulting controllers correspond to dynamic \emph{output-feedback} controllers of the form
\begin{align*}
u_k=K_1^uu_{k-1}+\dots+K_n^uu_{k-n}+K_1^yy_{k-1}+\dots+K_n^yy_{k-n}
\end{align*}
with $K=\begin{bmatrix}K_n^u&\dots&K_1^u&K_n^y&\dots&K_1^y\end{bmatrix}$ for~\eqref{eq:sys_IO}.
As in the previous sections, one can also include various forms of additional prior knowledge to reduce conservatism.

It is important to point out that this approach is not always applicable.
If no prior knowledge is available, then it requires that either $p=1$ or, if $p>1$, a certain condition involving the system's \emph{lag} $l$ holds (i.e., $pl=n$), compare~\cite[Lemma 13]{berberich2021on} for details.
It is an interesting issue for future research to overcome this limitation for robust output-feedback design within the present framework.
Finally, we note that a similar approach to data-driven output-feedback which is also based on an extended state vector but does not exploit the structure in~\eqref{eq:sys_IO_state_space} is suggested in~\cite{persis2020formulas}.


\section{Prior knowledge descriptions}\label{sec:prior}
In this section, we showcase the flexibility of the proposed framework by discussing various types of prior knowledge that can be considered.
In Sections~\ref{subsec:prior_uncertainty} and \ref{subsec:prior_disturbance},
we address prior uncertainty descriptions $\Dep$ (Assumption~\ref{ass:prior_multipliers})
and disturbance descriptions $\bm{D}$ (Assumption~\ref{ass:noise_multipliers}), respectively.

\subsection{Prior knowledge on the uncertainty (Assumption~\ref{ass:prior_multipliers})}\label{subsec:prior_uncertainty}
Suppose $\Delta=\mathrm{diag}_{j=1}^{\ell}(\Delta_j)$  and
$\Delta_j$ for $j=1,\dots,\ell$ satisfy
\begin{align}\label{eq:prior_bounds}
\begin{bmatrix}\Delta_j^\top\\I\end{bmatrix}^\top H_j\begin{bmatrix}\Delta_j^\top\\I\end{bmatrix}\succeq0\quad\text{or}\quad\begin{bmatrix}\delta_j\\1\end{bmatrix}^\top h_j\begin{bmatrix}\delta_j\\1\end{bmatrix}\geq0
\end{align}
with some $H_j\in\mathbb{R}^{(n_{w,i}+n_{z,i})\times(n_{w,i}+n_{z,i})}$ or $h_j\in\mathbb{R}^{2\times 2}$ if $\Delta_j$ is a full block or a repeated scalar block $\Delta_j=\delta_jI_{n_{w,i}}$, respectively.
Then, $\bm{P}_j$ in Assumption~\ref{ass:prior_multipliers} can be chosen as
\begin{align}\label{eq:prior_bounds_multipliers}
\bm{P}_{\mathrm{full}}&=\{P_j\mid P_j=\lambda H_j,\>\lambda\geq0\}\\\nonumber
\text{or}\>\>\bm{P}_{\mathrm{rep}}&=\{P_j\mid P_j=h_j\otimes \Lambda,\>0\preceq\Lambda\in\mathbb{R}^{n_{w,j}\times n_{w,j}}\},
\end{align}
respectively.
Important special cases are norm bounds imposed as $\Delta_j\Delta_j^\top\preceq\bar{\delta}I$ or $\delta_j^2\leq\bar{\delta}$ with some $\bar{\delta}>0$, which are captured by~\eqref{eq:prior_bounds} with
\begin{align}\label{eq:prior_bounds_Hj_hj}
    H_j=\begin{bmatrix}-I&0\\0&\bar{\delta} I\end{bmatrix}\quad\text{or}\quad
    h_j=\begin{bmatrix}-1&0\\0&\bar{\delta}\end{bmatrix}.
\end{align}
When considering $H_j$, $h_j$ in~\eqref{eq:prior_bounds_Hj_hj}, the multiplier class $\bm{P}_{\mathrm{rep}}$ for repeated scalar uncertainties is \emph{larger} than the class $\bm{P}_{\mathrm{full}}$, i.e., $\bm{P}_{\mathrm{full}}\subseteq\bm{P}_{\mathrm{rep}}$.
Generally speaking, a larger set of multipliers $\bm{P}_j$ shrinks the corresponding set of uncertainties $\pmb{\Delta}_j$, which is, typically, beneficial for robust controller design, at the price of increasing the computational complexity.
We note that~\eqref{eq:prior_set_Delta} also allows for multiplier descriptions that are more flexible than~\eqref{eq:prior_bounds_multipliers} such as convex hull multipliers (see, e.g.,~\cite[Section 2.2]{scherer2005relaxations}) in case $\Delta_j$ lies in the convex hull of a given set of generators. 
Let us include the following observation.

\begin{proposition}\label{prop:diagonal_delta}
If $\Delta_j\in\bm{\Delta}_j$ holds with $\bm{P}_j=\bm{P}_{\mathrm{rep}}$ for some $h_j=\begin{bmatrix}h_{11}&h_{12}\\h_{12}&h_{22}\end{bmatrix}$ with $h_{11}<0$, then $\Delta_j$ is diagonally repeated, i.e., it is of the form $\Delta_j=\delta_j I$ for some $\delta_j\in\mathbb{R}$.
\end{proposition}
\begin{proof}
We first show that $\Delta_j$ is diagonal.
Plugging $P_j=h_j\otimes\Lambda$ with $\Lambda=e_ke_k^\top\succeq0$
and the $k$-th unit vector $e_k$ into
\begin{align}\label{eq:prop_diagonal_delta_proof1}
    \begin{bmatrix}\Delta_j^\top\\I\end{bmatrix}^\top P_j
    \begin{bmatrix}\Delta_j^\top\\I\end{bmatrix}\succeq0,
\end{align}
we obtain
\begin{align}\label{eq:prop_diagonal_delta_proof2}
h_{11}\Delta_je_ke_k^\top\Delta_j^\top
+h_{12}(\Delta_je_ke_k^\top+e_ke_k^\top\Delta_j^\top)
+h_{22}e_ke_k^\top\succeq0
\end{align}
for $k=1,\dots,n_{w,j}$.
For any $i$ and $k$ with $i\neq k$, multiplying~\eqref{eq:prop_diagonal_delta_proof2} by $e_i^\top$ and $e_i$ from left and right and using $e_i^\top e_k=0$, we obtain $h_{11}(e_i^\top\Delta_je_k)^2\geq0$; since
$h_{11}<0$, this implies $e_i^\top\Delta_je_k=0$, i.e., $\Delta_j$ is diagonal.
Hence, $\Delta_j=\mathrm{diag}_{i=1}^{n_{w,j}}(\delta_{j,i})$ for some $\delta_{j,i}\in\mathbb{R}$.
Further, let $i,k\in\{1,\dots,n_{w,j}\}$ be arbitrary.
Then, choosing $P_j=h_j\otimes\Lambda$ with $\Lambda=(e_i-e_k)(e_i-e_k)^\top\succeq0$ and multiplying~\eqref{eq:prop_diagonal_delta_proof1} from left and right by $(e_i+e_k)^\top$ and $e_i+e_k$, respectively, leads to $h_{11}(\delta_{j,i}-\delta_{j,k})^2\geq0$, i.e., $\delta_{j,i}=\delta_{j,k}$.
\end{proof}

Proposition~\ref{prop:diagonal_delta} shows that any full-block uncertainty which satisfies $\Delta_j\in\bm{\Delta}_j$ for the multiplier class $\bm{P}_{\mathrm{rep}}$ is \emph{necessarily} of the repeated diagonal form $\Delta_j=\delta_jI$ for some $\delta_j$.
Hence, we can focus on unstructured uncertainty blocks $\Delta_j$ throughout the paper since any knowledge about a repeated diagonal structure of some $\Delta_j$ can be incorporated via an appropriate choice of $\bm{P}_j$ in Assumption~\ref{ass:prior_multipliers}.

Assuming negative definiteness of the left-upper block of some $P_j\in\bm{P}_j$ in Assumption~\ref{ass:prior_multipliers} is crucial in order to handle prior knowledge on the uncertainty structure (cf.\ Proposition~\ref{prop:diagonal_delta} and the results in Section~\ref{subsec:param_prior}).
This condition is, e.g., satisfied in the common case of defining the multipliers $\bm{P}_{\mathrm{full}}$ or $\bm{P}_{\mathrm{rep}}$ based on norm bounds, cf.~\eqref{eq:prior_bounds_Hj_hj}.

It is important to point out that the bounds in~\eqref{eq:prior_set_Delta_i} take a ``dual'' form involving $\Delta_j^\top$ instead of $\Delta_j$, the latter being more common in robust control.
For many practical uncertainty descriptions such as~\eqref{eq:prior_bounds}, it is possible under mild technical assumptions (if $H_j$ has $n_{w,j}$ positive and $n_{z,j}$ negative eigenvalues) to transform
bounds on $\Delta_j^\top$ into bounds on $\Delta_j$ (and vice versa) using the dualization lemma~\cite[Lemma 4.9]{scherer2000linear}.
As we have seen in Section~\ref{subsec:param_data}, such ``dual'' descriptions emerge when learning multipliers from data, which is why we work with bounds on $\Delta^\top$ instead of $\Delta$ throughout the paper.

\subsection{Prior knowledge on the disturbance (Assumption~\ref{ass:noise_multipliers})}\label{subsec:prior_disturbance}
Let us now illustrate the flexibility of modeling prior knowledge on the disturbance via Assumption~\ref{ass:noise_multipliers} by discussing various important special cases.
\subsubsection*{Quadratic full-block bounds}
Assumption~\ref{ass:noise_multipliers} contains knowledge of bounds such as
\begin{align}\label{eq:noise_bound_quadratic}
    \begin{bmatrix}D_{\tr}^\top\\I\end{bmatrix}^\top \begin{bmatrix}Q_d&S_d\\S_d^\top&R_d\end{bmatrix}\begin{bmatrix}D_{\tr}^\top\\I\end{bmatrix}\succeq0
\end{align}
with $Q_d\prec0$ as a simple special case when choosing $\bm{P}_d$ as
\begin{align}\label{eq:noise_multiplier_quadratic}
    \bm{P}_{\mathrm{quad}}=\left\{\lambda \begin{bmatrix}Q_d&S_d\\S_d^\top&R_d\end{bmatrix}\Bigm| \lambda>0\right\}.
\end{align}
The constraint~\eqref{eq:noise_bound_quadratic} encompasses various relevant scenarios:
\begin{enumerate}
\item A bound on the maximal singular value $\sigma_{\max}(D_{\tr})\leq\bar{d}_{2}$ or a norm bound on the sequence $\{\hat{d}_k\}_{k=0}^{N-1}$ such as $\lVert \hat{d}\rVert_2\leq\bar{d}_2$, both of which correspond to $Q_d=-I$, $S_d=0$, $R_d=\bar{d}_{2}^2I$.
\item A pointwise-in-time norm bound on the sequence $\{\hat{d}_k\}_{k=0}^{N-1}$, i.e., $\lVert\hat{d}_k\rVert_2\leq\bar{d}_{\infty}$, which corresponds to $Q_d=-I$, $S_d=0$, $R_d=\bar{d}_{\infty}^2NI$.
\item Further, inspired by~\cite{dandrea2001convex}, the above description also allows us to constrain the disturbance to be (approximately) contained in the kernel
of some (Toeplitz) matrix $T$, which might correspond to an LTI system:
For a scalar $\varepsilon\geq0$, the condition $DTT^\top D^\top\preceq\varepsilon I$ is of the form~\eqref{eq:noise_bound_quadratic} with $Q_d=-TT^\top,S_d=0,R_d=\varepsilon I$.
For instance, if it is known that $\{\hat{d}_k\}_{k=0}^{N-1}$ is constant, then $\varepsilon$ may be chosen as zero and $T$ may be taken to be
\begin{align*}
T=\begin{bmatrix}
-1&0&\dots&0\\
1&-1&\ddots&\vdots\\
0&\ddots&\ddots&0\\
\vdots&\ddots&1&-1\\
0&\dots&0&1\end{bmatrix}.
\end{align*}
Note that, in this case, $Q_d$ is only negative \emph{semi}-definite.
However, $\begin{bmatrix}I&0\end{bmatrix}P_d\begin{bmatrix}I&0\end{bmatrix}^\top\prec0$ for some $P_d\in\bm{P}_d$ can still be ensured, e.g., if the above description is combined with a bound as in 1) or 2) (see below for details).
Similarly, it is possible to construct Toeplitz matrices $T$ as above if $\{\hat{d}_k\}_{k=0}^{N-1}$ is periodic or, more generally, if it is generated by an LTI system.
\end{enumerate}

Quadratic full-block bounds as in~\eqref{eq:noise_bound_quadratic} have been used extensively in the recent literature on data-driven control~\cite{koch2022provably,persis2020formulas,berberich2020design,waarde2022from}.
The bound was employed for data-driven dissipativity analysis and controller design in~\cite{koch2022provably,waarde2022from}.
Under mild inertia assumptions on the inner matrix (cf.~\cite[Lemma 4.9]{scherer2000linear}),~\eqref{eq:noise_bound_quadratic} is equivalent to a bound on the (non-transposed) disturbance $D_{\tr}$ which was first considered in~\cite{persis2020formulas} and later generalized in~\cite{berberich2020design}.


\subsubsection*{Diagonal multipliers}
Let us assume that the disturbance satisfies a pointwise-in-time Euclidean norm bound of the form
\begin{align}\label{eq:noise_bound_pointwise_2}
\lVert\hat{d}_k\rVert_{2}\leq\bar{d}_2\>\>\text{for all}\>\>k=0,\dots,N-1
\end{align}
for some $\bar{d}_2>0$.
Then, the following is a valid class of multipliers in the sense of Assumption~\ref{ass:noise_multipliers}:
\begin{align}\label{eq:noise_multiplier_diagonal}
    \bm{P}_{\mathrm{diag}}=\left\{\begin{bmatrix}-\mathrm{diag}_{i=1}^N(\lambda_i)&0\\0&\sum_{i=1}^N\lambda_i\bar{d}_2^2I_{n_d}\end{bmatrix}\Bigm|
    \lambda_i\geq0\>\>\forall i\right\}.
\end{align}
The class~\eqref{eq:noise_multiplier_diagonal} includes one scalar multiplier $\lambda_i\geq0$ for each data point, whereas~\eqref{eq:noise_multiplier_quadratic} involves only one scalar multiplier for the full disturbance matrix.
We note that~\cite{martin2021dissipativity} and~\cite{bisoffi2021trade} use a similar approach with one multiplier per data point for dissipativity analysis of polynomial systems and stabilizing controller design for linear systems, respectively.

\subsubsection*{Convex hull multipliers}
Suppose the true disturbance $D_{\tr}$ lies in the convex hull generated by a set of \emph{known} matrices $\{\bar{D}_1,\dots,\bar{D}_{n_{c}}\}$.
This is, e.g., the case if the disturbance satisfies an infinity-norm bound
\begin{align}\label{eq:noise_bound_pointwise}
\lVert\hat{d}_k\rVert_{\infty}\leq\bar{d}_{\infty}\>\>\text{for all}\>\>k=0,\dots,N-1
\end{align}
for some $\bar{d}_{\infty}>0$.
Then, Assumption~\ref{ass:noise_multipliers} holds with $\bm{P}_d$ as
\begin{align}\label{eq:noise_multiplier_convex_hull}
    \bm{P}_{\mathrm{con}}=\left\{P_d=P_d^\top\Bigm|\begin{bmatrix}I\\0\end{bmatrix}^\top\!\!\! P_d\begin{bmatrix}\star\end{bmatrix}\preceq0,
    \begin{bmatrix}\bar{D}_i^\top\\I\end{bmatrix}^\top \!\!\!P_d\begin{bmatrix}\star\end{bmatrix}\succeq0\>\forall i\right\}\!.
\end{align}
Indeed, suppose $D_{\tr}=\sum_{i=1}^{n_{c}}\lambda_i\bar{D}_i$ for some $\lambda_i\geq0$ with $\sum_{i=1}^{n_{c}}\lambda_i=1$ and choose an arbitrary $P_d\in\bm{P}_{\mathrm{con}}$.
Since the left-upper block of $P_d$ is negative definite, the mapping $D\mapsto F(D)\coloneqq\begin{bmatrix}D^\top\\I\end{bmatrix}^\top \!\!\!P_d\begin{bmatrix}D\\I\end{bmatrix}$ is concave, which implies
\begin{align*}
    F(D_{\tr})=F\left(\sum_{i=1}^{n_{c}}\lambda_i\bar{D}_i\right)\succeq\sum_{i=1}^{n_{c}}\lambda_iF(\bar{D}_i)\stackrel{\eqref{eq:noise_multiplier_convex_hull}}{\succeq}0.
\end{align*}
\subsubsection*{Further multipliers and combinations}
As a key advantage of the proposed description, any two valid convex cones of (non-redundant) multipliers $\bm{P}_d^1$ and $\bm{P}_d^2$ can be combined as $\bm{P}_d=\bm{P}_d^1+\bm{P}_d^2$ in order to tighten the disturbance description.
For instance, one can combine diagonal multipliers $\bm{P}_d^1$ and convex hull multipliers $\bm{P}_d^2$ for a subset of the available data in order to potentially reduce conservatism.
Another noteworthy example is the combination of bounds as in~\eqref{eq:noise_bound_quadratic},~\eqref{eq:noise_bound_pointwise_2}, or~\eqref{eq:noise_bound_pointwise} with information on periodicity.

Finally, also in the context of Assumption~\ref{ass:noise_multipliers}, we note that one can work with more general classes of multipliers from the robust control literature, such as those based on Lagrange~\cite{scherer2006lmi}
or SOS~\cite[Section 6]{scherer2006matrix} relaxations, employing either implicit or explicit set descriptions of the disturbance bound.

\subsubsection*{Summary}
In the following, we discuss the relations among the above multiplier classes in order to illustrate their potentials for data-driven control.
If the disturbance satisfies an infinity-norm bound as in~\eqref{eq:noise_bound_pointwise}, then the convex hull multipliers $\bm{P}_{\mathrm{con}}$ provide the tightest and thus most powerful description among the three classes $\bm{P}_{\mathrm{quad}}$, $\bm{P}_{\mathrm{diag}}$, and $\bm{P}_{\mathrm{con}}$.
More precisely, when defining $\bm{P}_{\mathrm{quad}}$ as in~\eqref{eq:noise_multiplier_quadratic} via $Q_d=-I$, $S_d=0$, $R_d=n_d\bar{d}_{\infty}^2N$ (with $n_d$ resulting from $\lVert d\rVert_{2}\leq\sqrt{n_d}\lVert d\rVert_{\infty}$ for any $d\in{\Bbb R}^{n_d}$) as well as $\bm{P}_{\mathrm{diag}}$ as in~\eqref{eq:noise_multiplier_diagonal} via $\bar{d}_2=\sqrt{n_d}\bar{d}_{\infty}$ and $\bm{P}_{\mathrm{con}}$ as in~\eqref{eq:noise_multiplier_convex_hull}, then all three are valid according to Assumption~\ref{ass:noise_multipliers} and satisfy $\bm{P}_{\mathrm{quad}}\subseteq\bm{P}_{\mathrm{diag}}\subseteq\bm{P}_{\mathrm{con}}$.

Additionally, $\bm{P}_{\mathrm{diag}}$ and $\bm{P}_{\mathrm{con}}$ have a crucial qualitative advantage over $\bm{P}_{\mathrm{quad}}$.
In fact, if we use data affected by a disturbance with pointwise bounds such as~\eqref{eq:noise_bound_pointwise_2} or~\eqref{eq:noise_bound_pointwise} in our design approach, then the resulting closed-loop performance cannot deteriorate with $\bm{P}_{\mathrm{diag}}$ and $\bm{P}_{\mathrm{con}}$ if additional data points are added.
As we will see later in the paper, this is not the case when translating~\eqref{eq:noise_bound_pointwise_2} into~\eqref{eq:noise_bound_quadratic} via $Q_d=-I$, $S_d=0$, $R_d=\bar{d}^2NI$.
This is an  important insight since pointwise disturbance bounds are arguably more realistic and useful for practical applications if compared to the quadratic full-block bound~\eqref{eq:noise_bound_quadratic}.
On the other hand, implementing $\bm{P}_{\mathrm{quad}}$ only involves one scalar decision variable, whereas $\bm{P}_{\mathrm{diag}}$ and $\bm{P}_{\mathrm{con}}$ involve $N$ and $2^{n_dN}$ decision variables, respectively.
Thus, while $\bm{P}_{\mathrm{diag}}$ is applicable for medium-sized problems (cf.\ Section~\ref{sec:example}), $\bm{P}_{\mathrm{con}}$ can only be used for systems with low spatial disturbance dimension and few data points.

\begin{remark}
Let us emphasize that all results in this paper can be trivially extended in case that multiple trajectories of~\eqref{eq:sys} are available, simply by stacking them together, even if the \emph{concatenated} sequences do not constitute trajectories of~\eqref{eq:sys}.
In fact, the presented methods are applicable for arbitrary matrices $\tilde{X}_+$, $\tilde{X}$, $\tilde{U}$, possibly composed of several trajectory pieces, as long as they satisfy the data equation
\begin{align}
\tilde{X}_+=A\tilde{X}+B\tilde{U}+B_w\Delta_{\tr}(C_z\tilde{X}+D_z\tilde{U})+B_d\tilde{D}
\end{align}
for some $\tilde{D}\in\bm{D}$.
This observation is important, e.g., when dealing with unstable systems for which the generation of a reasonably long data trajectory can be difficult.
\end{remark}

\section{Towards tight design conditions}\label{sec:necessity}

Theorem~\ref{thm:robust_H2} and all further results proposed in this paper rely on a common quadratic Lyapunov function (CQLF).
In general, our results are not necessary for robust stability and performance certified with a CQLF for prior uncertainties and disturbance bounds. 
In the following, we discuss possibilities to reduce conservatism as well as conditions under which our design approach is indeed tight.

First, starting from the uncertainty description obtained via Lemma~\ref{lem:combined_multipliers}, one can further refine the representation by subsequently including additional multipliers for the individual components $\Delta_j$, for the full uncertainty $\Delta_{\tr}$, or for the transformed full uncertainty $\tDet$, and combinations thereof.
For instance, suppose we have convex cones of multipliers $\bm{P}_f$, $\bm{\tilde{P}}_f$ such that, for any $P_f\in\bm{P}_f$, $\tilde{P}_f\in\bm{\tilde{P}}_f$, it holds that
\begin{align}\label{eq:full_block_multipliers_for_Delta}
    \begin{bmatrix}\Delta_{\tr}^\top\\I\end{bmatrix}^\top P_f\begin{bmatrix}\Delta_{\tr}^\top\\I\end{bmatrix}\succeq0\>\>
    \text{and}\>\>
    \begin{bmatrix}\tDet^\top\\I\end{bmatrix}^\top \tilde{P}_f\begin{bmatrix}\tDet^\top\\I\end{bmatrix}\succeq0.
\end{align}
Then, a valid class of multipliers for $\tilde{\Delta}_{\tr}$ is obtained via
\begin{align}\label{eq:mult_extended}
    \bm{\tilde{P}}_{\mathrm{com}}+\begin{bmatrix}I&0\\0&B_w^\top\end{bmatrix}^\top\bm{P}_f
    \begin{bmatrix}I&0\\0&B_w^\top\end{bmatrix}
    +\bm{\tilde{P}}_f.
\end{align}
%
Here, $\bm{P}_f$ and $\bf{\tilde{P}}_f$ can be chosen, e.g., based on the S-procedure, a convex hull, Lagrange relaxations, or SOS relaxations, analogous to the multipliers for describing $\Dep$ and $\bf{D}$ (compare Section~\ref{sec:prior}).
All these refinements allow to gradually reduce conservatism and improve the guaranteed performance obtained via Theorem~\ref{thm:robust_H2}, simply by replacing $\bm{\tilde{P}}_{\mathrm{com}}$ in~\eqref{eq:prop_rob_H2_LMI2} by a larger class of multipliers as in~\eqref{eq:mult_extended}.

Further, there are various cases in which our framework provides necessary and sufficient conditions for robust stability and performance with a CQLF.
First, our results are tight if no prior knowledge on $\Det$ is available, i.e., $\Dep=\mathbb{R}^{n_w\times n_z}$, and if the disturbance is captured by~\eqref{eq:noise_multiplier_quadratic}.
In this case, the S-procedure provides sufficient \emph{and} necessary conditions (under a suitable constraint qualification), compare, e.g.,~\cite{waarde2022from}.

Since \eqref{eq:sys_prior_tilde} actually depends affinely on $\tilde{\Delta}$, our results are also necessary for robust stability and performance with a CQLF if $\tDecom$ is a (compact) polytope.
Indeed,
the full-block S-procedure~\cite{scherer2001lpv}
shows that \eqref{eq:prop_rob_H2_proof1} holds for all $\tilde\Delta\in\tDecom$ iff
there exists some symmetric $\tilde{P}_{\mathrm{com}}$ which satisfies~\eqref{eq:prop_rob_H2_LMI2} and
\begin{align}\label{eq:full_block}
    \begin{bmatrix}\tilde{\Delta}^\top\\I\end{bmatrix}^\top
    \tilde{P}_{\mathrm{com}}
    \begin{bmatrix}\tilde{\Delta}^\top\\I\end{bmatrix}\succeq0\>\>\text{for all}\>\>\tilde{\Delta}\in\tDecom.
\end{align}
By inspection, \eqref{eq:prop_rob_H2_LMI2} implies that the left-upper block of $\tilde{P}_{\mathrm{com}}$ is negative definite.
Therefore, a standard concavity argument shows that $\tilde{P}_{\mathrm{com}}$ just is, actually, a convex
hull multiplier for the polytope $\tDecom$.

These arguments apply if $\Delta$ is a full unstructured uncertainty and either i) no data but only prior bounds or ii) no prior bounds and only data with $Z$ of full row rank is available, and if the prior uncertainty or disturbance bounds are expressed by polytopes; then convex hull multipliers in Assumption~\ref{ass:prior_multipliers} or~\ref{ass:noise_multipliers} (cf.\ Section~\ref{sec:prior}) lead to tight conditions for controller design.

More generally, if both $\Dep$ and the disturbance set $\bm{D}$ are polytopes, then
$\tDep\cap\tDel$ is as well a polytope for which one can determine an explicit representation
as the convex hull of finitely many matrices. Using convex hull multipliers
in robust controller synthesis for~\eqref{eq:sys_prior_tilde} also  generates tight results.

In case that $\Dep$ and $\bm{D}$ are non-polytopic but have an LMI representation or are even just semi-algebraic (i.e., defined via polynomial inequalities), we emphasize that suitable multiplier classes can be constructed by SOS relaxation techniques.
The results in \cite{scherer2006matrix} open the way for constructing relaxation families which are asymptotically exact, but
typically at the cost of large computational complexity.
Quantifying the relaxation gap for different choices of multiplier classes for both $\Dep$ and $\bm{D}$ remains to be explored.

Finally, the fact that we work with CQLFs is an additional source of conservatism of our framework.
Reducing conservatism and improving performance via parameter-dependent Lyapunov functions is an interesting issue for future research.
It is worth noting that the main sources of conservatism of our framework are inherited from classical multiplier relaxations in robust control (respecting structure, considering a CQLF, etc.), but we do not introduce additional discrepancies.

\section{Extension to nonlinear uncertainties}\label{sec:nonlinear}
\begin{subequations}\label{eq:sys_nl}

In the following, we illustrate the flexibility of our framework by considering additional and more general uncertainties such as static nonlinearities.
We consider systems of the form
\begin{align}\label{eq:sys_nl1}
\left[
\begin{array}{c}x_{k+1}\\\hline e_k\\z_k\\z_k'\end{array}
\right]
&=\left[
\begin{array}{c|cccc}
A&B&B_d&B_w&B_w'\\\hline
C_e&D_{eu}&D_{ed}&0&0\\
C_z&D_{z}&0&0&D_{zw}'\\
C_z'&D_{z}'&0&0&0\end{array}\right]
\left[
\begin{array}{c}
x_k\\\hline u_k\\d_k\\w_k\\w_k'
\end{array}\right]\\\label{eq:sys_nl2}
w_k&=\Delta z_k,\>\>
w_k'=\Delta'(x_k)z_k',
\end{align}
\end{subequations}
which, if compared to~\eqref{eq:sys_prior}, have an additional uncertainty channel with variables $w_k'\in\mathbb{R}^{n_w'}$ and $z_k'\in\mathbb{R}^{n_z'}$, and with known matrices $B_w'$, $D_{zw}'$, $C_z'$, $D_{z}'$.
Moreover, $\Delta':\mathbb{R}^{n}\to\mathbb{R}^{n_w'\times n_z'}$ is a \emph{nonlinear} mapping satisfying $\Delta'(0)=0$.
This is mainly assumed for simplicity of notation. All of the following results can be extended to time-varying
or dynamic structured uncertainties $\Delta'$ after suitable adaptions and by working with \emph{integral} quadratic constraints as surveyed in~\cite{scherer2021dissipativity}.

\begin{assumption}\label{ass:nl_multipliers}
There exists a convex cone of symmetric matrices $\bm{P'}$ admitting an LMI representation such that, for any $P'\in\bm{P'}$, it holds that
\begin{align}\label{eq:ass_nl_multipliers}
    \begin{bmatrix}\Delta'(x)^\top\\I\end{bmatrix}^\top P'
    \begin{bmatrix}\Delta'(x)^\top\\I\end{bmatrix}\succeq0\quad
    \text{for all}\>\>x\in\mathbb{R}^{n}.
\end{align}
\end{assumption}
Multipliers $\bm{P'}$ as in Assumption~\ref{ass:nl_multipliers} are available if it is known that the norm of $\Delta'$ is bounded by $\gamma$, leading to
\begin{align*}
    \bm{P'}=\left\{P'\Bigm| P'=\lambda\begin{bmatrix}-I_{n_z'}&0\\0&\gamma^2 I_{n_w'}\end{bmatrix},\>\lambda>0\right\}.
\end{align*}
Further special cases include, e.g., sector or slope bounds, compare~\cite{veenman2016robust} for details.
We now show that our framework can handle nonlinear uncertainties $\Delta'$ as in~\eqref{eq:sys_nl}, provided that they satisfy Assumption~\ref{ass:nl_multipliers} and data of $\Delta'$ are available.
To this end, suppose that data $\{x_k\}_{k=0}^{N}$, $\{u_k,\Delta'(x_k)\}_{k=0}^{N-1}$ of~\eqref{eq:sys_nl} with some unknown \emph{true} uncertainty $\Delta=\Delta_{\tr}$ are available, which are affected by an unknown disturbance sequence $\{\hat{d}_k\}_{k=0}^{N-1}$ satisfying Assumption~\ref{ass:noise_multipliers}.
This allows us to determine $\{z_k'\}_{k=0}^{N-1}$ and hence $\{w_k'\}_{k=0}^{N-1}$.
We can then proceed as in Section~\ref{subsec:param_data} if $\Delta'$ is treated as a \emph{known} component of the LFT~\eqref{eq:sys_nl} and we learn multipliers only for the parametric uncertainty $\Det$.
To be precise, the set of uncertainties $\Delta$ consistent with the data can be defined exactly as in~\eqref{eq:Sigma_Delta_def}, with the only difference that the definitions of $M$ and $Z$ in~\eqref{eq:MZ_def} are modified to
\begin{align}\label{eq:MZ_def_nl}
M&\coloneqq X_+-AX-BU-B_w'W',\\\nonumber
Z&\coloneqq C_zX+D_zU+D_{zw}'W'.
\end{align}
Lemma~\ref{lem:data_bound} remains true for these matrices $M$ and $Z$ and the corresponding (modified) definition of $\tDel$ based on~\eqref{eq:data_multipliers} with $M$ and $Z$ as in~\eqref{eq:MZ_def_nl}.
Hence, we can combine prior knowledge with data of~\eqref{eq:sys_nl} to learn a class of multipliers $\bm{\tilde{P}}_{\mathrm{com}}$ for the uncertainty $\Delta_{\tr}$, i.e., $B_w\Delta_{\tr}\in\tDecom$ with $\tDecom$ as in~\eqref{eq:tDecom_def}.
In doing so, the nonlinear uncertainty $\Delta'$ does not pose any additional difficulties.
We now employ the multipliers $\bm{\tilde{P}}_{\mathrm{com}}$ and $\bm{P'}$ for $B_w\Delta$ and $\Delta'$, respectively, to design state-feedback controllers which guarantee robust quadratic performance of~\eqref{eq:sys_nl}.
\begin{definition}\label{def:quad_perf}
We say that the system~\eqref{eq:sys_nl} under state-feedback $u_k=Kx_k$ satisfies robust quadratic performance with index $P_p=\begin{bmatrix}Q_p&S_p\\S_p^\top&R_p\end{bmatrix}$, if there exists an $\varepsilon>0$ such that
\begin{align}\label{eq:quad_perf}
\sum_{k=0}^\infty \begin{bmatrix}d_k\\e_k\end{bmatrix}^\top
P_p\begin{bmatrix}d_k\\e_k\end{bmatrix}\leq-\varepsilon\sum_{k=0}^\infty {d_k}^\top d_k
\end{align}
for all $d\in\ell_2$ and all uncertainties $\Delta\in\Decom$ and $\Delta'$ satisfying Assumption~\ref{ass:nl_multipliers}.
\end{definition}
Definition~\ref{def:quad_perf} includes standard performance specifications such as, e.g., a bound on the $\mathcal{L}_2$-gain from $d$ to $e$, which corresponds to the choice $Q_p=-\gamma^2 I$, $S_p=0$, $R_p=I$, or strict passivity of the channel $d\mapsto e$, which corresponds to $Q_p=0$, $S_p=-I$, $R_p=0$.
We assume that $P_p$ is invertible and we define $\begin{bmatrix}\tilde{Q}_p&\tilde{S}_p\\\tilde{S}_p^\top&\tilde{R}_p\end{bmatrix}
\coloneqq P_p^{-1}$ as well as $\tilde{P}_p\coloneqq\begin{bmatrix}-\tilde{R}_p&\tilde{S}_p^\top\\\tilde{S}_p&-\tilde{Q}_p\end{bmatrix}$.
Further, we assume $\tilde{Q}_p\preceq0$, which holds for the most common performance specifications as those above.
The following result provides a design procedure for state-feedback controllers with guaranteed robust quadratic performance of~\eqref{eq:sys_nl}.

\begin{theorem}\label{thm:robust_perf_nl}
If Assumptions~\ref{ass:prior_multipliers}--\ref{ass:nl_multipliers} hold and there exist ${\mathcal{Y}\succ0}$, $K$, $\tilde{P}_{\mathrm{com}}\in\bm{\tilde{P}}_{\mathrm{com}}$, and $P'\in\bm{P'}$
such that~\eqref{eq:prop_robust_perf_nl_LMI} holds, then~\eqref{eq:sys_nl} with $u_k=Kx_k$ is stable and satisfies robust quadratic performance with index $P_p$.
\begin{figure*}
\vspace{0pt}
\begin{align}\label{eq:prop_robust_perf_nl_LMI}
\mleft[\begin{array}{cccc}I&0&0&0\\
(A+BK)^\top&(C_z'+D_z'K)^\top&(C_z+D_zK)^\top&(C_e+D_{eu}K)^\top\\\hline
0&I&0&0\\
{B_w'}^\top&0&{D_{zw}'}^\top&0\\\hline
0&0&I&0\\
I&0&0&0\\\hline
0&0&0&I\\
B_d^\top&0&0&D_{ed}^\top\end{array}\mright]^\top
\mleft[
\begin{array}{c|c|c|c}
\begin{matrix}-\mathcal{Y}&0\\0&\mathcal{Y}\end{matrix}
&\begin{matrix}0&0\\0&0\end{matrix}&\begin{matrix}0&0\\0&0\end{matrix}&\begin{matrix}0&0\\0&0\end{matrix}\\\hline
\begin{matrix}0&0\\0&0\end{matrix}&P'&\begin{matrix}0&0\\0&0\end{matrix}&\begin{matrix}0&0\\0&0\end{matrix}\\\hline
\begin{matrix}0&0\\0&0\end{matrix}&\begin{matrix}0&0\\0&0\end{matrix}&\tilde{P}_{\mathrm{com}}&\begin{matrix}0&0\\0&0\end{matrix}\\\hline
\begin{matrix}0&0\\0&0\end{matrix}&\begin{matrix}0&0\\0&0\end{matrix}&\begin{matrix}0&0\\0&0\end{matrix}&
\tilde{P}_p
\end{array}\mright]
\mleft[\begin{array}{c}\star\end{array}\mright]\prec0
\end{align}
\noindent\makebox[\linewidth]{\rule{\textwidth}{0.4pt}}
\end{figure*}
\end{theorem}
\begin{proof}
We first show that any $P'\in\bm{P'}$ with~\eqref{eq:prop_robust_perf_nl_LMI} also satisfies a useful inertia property.
If using $\Delta'(0)=0$ in~\eqref{eq:ass_nl_multipliers}, we infer $\begin{bmatrix}0&I_{n_w'}\end{bmatrix} P'\begin{bmatrix}0&I_{n_w'}\end{bmatrix}^\top\succeq0$.
After potentially perturbing $P'$ by $\varepsilon I$ for some small $\varepsilon>0$, we obtain $\begin{bmatrix}0&I_{n_w'}\end{bmatrix} P'\begin{bmatrix}0&I_{n_w'}\end{bmatrix}^\top\succ0$ without destroying~\eqref{eq:ass_nl_multipliers} or~\eqref{eq:prop_robust_perf_nl_LMI}.
On the other hand, the second diagonal block of~\eqref{eq:prop_robust_perf_nl_LMI} reads
\begin{align*}
    (C_z'+D_z'K)\mathcal{Y}(C_z'+D_z'K)^\top+
    \begin{bmatrix}I_{n_z'}&0\end{bmatrix}P'\begin{bmatrix}I_{n_z'}&0\end{bmatrix}^\top\prec0.
\end{align*}
Using $\mathcal{Y}\succ0$, this implies $\begin{bmatrix}I_{n_z'}&0\end{bmatrix}P'\begin{bmatrix}I_{n_z'}&0\end{bmatrix}^\top\prec0$.
By the min-max principle of Courant-Fischer, $P'$ has $n_z'$ negative and $n_w'$ positive eigenvalues and is, hence, invertible.
Similarly, due to $\tilde{Q}_p\preceq0$, feasibility of~\eqref{eq:prop_robust_perf_nl_LMI}
and invertibility of $\tilde{P}_p$, the matrix $\tilde{P}_p$ has $n_e$ negative and $n_d$ positive eigenvalues.

Using the full-block S-procedure~\cite{scherer2001lpv} together with the definition of $\tDecom$ in~\eqref{eq:tDecom_def},~\eqref{eq:prop_robust_perf_nl_LMI} implies that~\eqref{eq:prop_robust_perf_nl_proof1} holds for any $\tilde{\Delta}\in\tDecom$, if defining $\mathcal{A}(\tilde{\Delta})\coloneqq A+BK+\tilde{\Delta}(C_z+D_zK)$.
\begin{figure*}
\vspace{0pt}
\begin{align}\label{eq:prop_robust_perf_nl_proof1}
\mleft[\begin{array}{ccc}I&0&0\\
\mathcal{A}(\tilde{\Delta})^\top&(C_z'+D_z'K)^\top&(C_e+D_{eu}K)^\top\\\hline
0&I&0\\
(B_w'+\tilde{\Delta}D_{zw}')^\top&0&0\\\hline
0&0&I\\
B_d^\top&0&D_{ed}^\top\end{array}\mright]^\top
\mleft[
\begin{array}{c|c|c}
\begin{matrix}-\mathcal{Y}&0\\0&\mathcal{Y}\end{matrix}
&\begin{matrix}0&0\\0&0\end{matrix}&\begin{matrix}0\>\>&\>\>0\\0\>\>&\>\>0\end{matrix}\\\hline
\begin{matrix}0&0\\0&0\end{matrix}&P'&\begin{matrix}0\>\>&\>\>0\\0\>\>&\>\>0\end{matrix}\\\hline
\begin{matrix}0&0\\0&0\end{matrix}&\begin{matrix}0&0\\0&0\end{matrix}&
\begin{matrix}-\tilde{R}_p&\tilde{S}_p^\top\\\tilde{S}_p&-\tilde{Q}_p
\end{matrix}
\end{array}\mright]
\mleft[\begin{array}{ccc}\star&\star&\star\\\star&\star&\star\\\hline\star&\star&\star\\
\star&\star&\star\\\hline\star&\star&\star\\\star&\star&\star\end{array}\mright]\prec0
\end{align}
\noindent\makebox[\linewidth]{\rule{\textwidth}{0.4pt}}
\end{figure*}
Due to the above inertia properties of $P'$ and $\tilde{P}_p$, we can apply
the dualization lemma~\cite[Lemma 4.9]{scherer2000linear} to~\eqref{eq:prop_robust_perf_nl_proof1} and infer that $\mathcal{X}\coloneqq\mathcal{Y}^{-1}\succ0$ satisfies~\eqref{eq:prop_robust_perf_nl_proof2} for any $\tilde{\Delta}\in\Decom$, where
\begin{align*}
    \tilde{P}'\coloneqq\begin{bmatrix}0&-I\\I&0\end{bmatrix}^\top P\text{$'$}^{-1}\begin{bmatrix}0&-I\\I&0\end{bmatrix}^\top.
\end{align*}
\begin{figure*}
\vspace{0pt}
\begin{align}\label{eq:prop_robust_perf_nl_proof2}
\mleft[\begin{array}{ccc}I&0&0\\
\mathcal{A}(\tilde{\Delta})&B_w'+\tilde{\Delta}D_{zw}'&B_d\\\hline
0&I&0\\
C_z'+D_z'K&0&0\\\hline
0&0&I\\
C_e+D_{eu}K&0&D_{ed}\end{array}\mright]^\top
\mleft[
\begin{array}{c|c|c}
\begin{matrix}-\mathcal{X}&0\\0&\mathcal{X}\end{matrix}
&\begin{matrix}0&0\\0&0\end{matrix}&\begin{matrix}0\>\>&\>\>0\\0\>\>&\>\>0\end{matrix}\\\hline
\begin{matrix}0&0\\0&0\end{matrix}&\tilde{P}'&\begin{matrix}0\>\>&\>\>0\\0\>\>&\>\>0\end{matrix}\\\hline
\begin{matrix}0&0\\0&0\end{matrix}&\begin{matrix}0&0\\0&0\end{matrix}&
\begin{matrix}Q_p&S_p\\S_p^\top&R_p
\end{matrix}
\end{array}\mright]
\mleft[\begin{array}{ccc}\star&\star&\star\\\star&\star&\star\\\hline\star&\star&\star\\
\star&\star&\star\\\hline\star&\star&\star\\\star&\star&\star\end{array}\mright]\prec0
\end{align}
\noindent\makebox[\linewidth]{\rule{\textwidth}{0.4pt}}
\end{figure*}
Using~\cite[Theorem 10.4]{scherer2000robust} and $B_w\Decom\subseteq\tDecom$, we infer that~\eqref{eq:sys_nl} under $u_k=Kx_k$ is stable and satisfies robust quadratic performance with index $P_p$ for all $\Delta\in\Decom$ if
\begin{align}\label{eq:prop_robust_perf_nl_proof3}
    \begin{bmatrix}\Delta'(x)\\I\end{bmatrix}^\top
    \tilde{P}'
    \begin{bmatrix}\Delta'(x)\\I\end{bmatrix}\succeq0\quad
    \text{for all}\>\>x\in\mathbb{R}^{n}.
\end{align}
Using again the inertia property of $P'$, the dualization lemma~\cite[Lemma 4.9]{scherer2000linear} implies that~\eqref{eq:prop_robust_perf_nl_proof3} is equivalent to~\eqref{eq:ass_nl_multipliers}, which concludes the proof.
\end{proof}

Since the uncertainty description given by $\tDecom$ in~\eqref{eq:tDecom_def} is of a \emph{dual} form, involving $\tilde{\Delta}^\top$ instead of $\tilde{\Delta}$, Theorem~\ref{thm:robust_perf_nl} formulates a design condition for the dual system.
The synthesis inequality~\eqref{eq:prop_robust_perf_nl_LMI} can also be interpreted as enforcing quadratic performance with index $\tilde{P}_p$ for the \emph{dual} system.
It is essential in the proof of Theorem~\ref{thm:robust_perf_nl} to not simply dualize~\eqref{eq:prop_robust_perf_nl_LMI} in order to arrive at a condition for the original system since this would require restrictive inertia assumptions on the set of multipliers $\bm{\tilde{P}}_{\mathrm{com}}$ and, thus, on the data.

Defining $L=K\mathcal{Y}$ and applying the Schur complement, one reformulates~\eqref{eq:prop_robust_perf_nl_LMI} as an LMI, exactly as in Section~\ref{subsec:design_H2}.
Thus, Theorem~\ref{thm:robust_perf_nl} provides a direct design procedure for controllers with guaranteed closed-loop quadratic performance for all parametric uncertainties $\Delta\in\Decom$ and all nonlinear uncertainties $\Delta'$ satisfying Assumption~\ref{ass:nl_multipliers}.
Note that performance is guaranteed as in~\eqref{eq:quad_perf}, i.e., over an \emph{infinite} time-horizon and for \emph{arbitrary} disturbance inputs $d$ not necessarily satisfying a bound such as $D\in\bm{D}$.

The recent papers~\cite{luppi2022on,waarde2021matrix} address data-driven control of LTI systems interconnected with static nonlinearities in case of noise-free data.
In contrast to our approach, these works consider only one multiplier which is fixed a priori to describe the nonlinearity, whereas optimizing over a whole class of multipliers $\bm{P'}$ as in Theorem~\ref{thm:robust_perf_nl} increases the flexibility and typically reduces the conservatism significantly.
Finally, Theorem~\ref{thm:robust_perf_nl} can also be seen as an alternative to recent results on data-driven control of linear parameter-varying systems~\cite{verhoek2021data}.

\section{Numerical examples}\label{sec:example}
In this section, we showcase the potential of our framework with two numerical examples.
First, we explore the influence of different forms of prior knowledge, disturbance multipliers, and data parameters on the closed-loop performance (Sections~\ref{subsec:ex_H2} and~\ref{subsec:ex_H2_multipliers}).
Next, in Section~\ref{subsec:ex_Hinf}, we exploit the possibility of including prior knowledge to perform data-driven $\mathcal{H}_{\infty}$-loop-shaping for the satellite system in Example~\ref{ex:introduction}.

\subsection{Robust $\mathcal{H}_2$-performance}
\label{subsec:ex_H2}%
We consider an academic example of the form
\begin{align}\label{eq:ex_H2}
x_{k+1}=\begin{bmatrix}0&0.5&-0.3\\
\delta_1&\Delta_{11}&\Delta_{12}\\
0.1&\Delta_{21}&\Delta_{22}\end{bmatrix}x_k+\begin{bmatrix}\delta_1\\1\\0.5
\end{bmatrix}u_k+d_k,
\end{align}
where $\Delta\coloneqq\mathrm{diag}_{j=1}^{2}(\Delta_j)$ for $\Delta_1\!=\!\delta_1I_2$ with
\emph{true} value $\delta_{1,\tr}=0.2$ and $\Delta_2=\begin{bmatrix}\Delta_{11}&\Delta_{12}\\\Delta_{21}&\Delta_{22}\end{bmatrix}$
with
$\Delta_{2,\tr}=\begin{bmatrix}0.5&-0.2\\-0.1&0.3\end{bmatrix}$.
It is straightforward to determine matrices $A$, $B$, $B_d$, $B_w$, $C_z$, and $D_z$ such that~\eqref{eq:ex_H2} is equivalent to~\eqref{eq:sys}.
We assume that a trajectory $\{x_k\}_{k=0}^{N}$, $\{u_k\}_{k=0}^{N-1}$ of the true system with length $N=200$ is available, which is generated by an input and disturbance sequence sampled uniformly as $u_k\in[-1,1]$ and satisfying $\hat{d}_k\in\{d\in\mathbb{R}^3\mid \lVert d\rVert_{\infty}\leq\bar{d}\}$ for some noise level $\bar{d}>0$, respectively.
For modeling the noise bound, we choose $Q_d=-I$, $S_d=0$, $R_d=\bar{d}^2n_dNI$ in~\eqref{eq:noise_bound_quadratic} 
and~\eqref{eq:noise_multiplier_quadratic}.
We consider four different scenarios for controller design:
\begin{enumerate}
\item All numerical values in~\eqref{eq:ex_H2} are known except for the values of $\delta_1$ and $\Delta_2$.
Moreover, uncertainty bounds of the form $\delta_{1,\tr}^2\leq0.1$ and $\Delta_{2,\tr}\Delta_{2,\tr}^\top\preceq0.5I$ are
assumed to be available and included by using the multiplier classes $\bm{P}_1$ and $\bm{P}_2$ as in~\eqref{eq:prior_bounds_multipliers} for $\Delta_1$ and $\Delta_2$, respectively.

\item All numerical values in~\eqref{eq:ex_H2} are unknown and no prior uncertainty bound on $\Delta$ is available, but the trajectory $\{x_k\}_{k=0}^{N}$, $\{u_k\}_{k=0}^{N-1}$ of length $N=200$ and generated according to the above specifications is known.

\item The combined information of 1) and 2) is available, i.e., all numerical values in~\eqref{eq:ex_H2} are known except for $\delta_1$ and $\Delta_2$, and we can use bounds as in 1) and data as in 2).

\item All numerical values in~\eqref{eq:ex_H2} are known, including the true values of $\delta_1$ and $\Delta_2$, i.e., the system~\eqref{eq:ex_H2} is known exactly without any uncertainty.
\end{enumerate}
For each of the four scenarios, we use Theorem~\ref{thm:robust_H2} to design a controller to minimize the respective bound on the $\mathcal{H}_2$-norm of the performance channel defined via $C_e=\begin{bmatrix}I\\0\end{bmatrix}$, $D_{eu}=\begin{bmatrix}0\\\frac{1}{5}I\end{bmatrix}$.
This corresponds to a linear-quadratic regulation problem with weight matrices $I$ and $\frac{1}{25}I$ for the state and input, respectively.
The optimal $\mathcal{H}_2$-norm bounds which can be guaranteed for each of the four scenarios
depending on the noise level are displayed in Figure~\ref{fig:example_H2}.

\begin{figure}
\begin{center}
\includegraphics[width=0.49\textwidth]{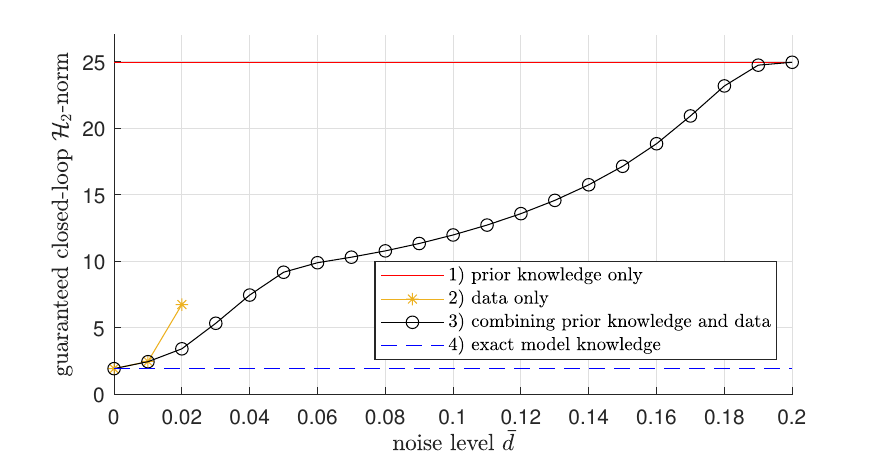}
\end{center}
\caption{Guaranteed closed-loop $\mathcal{H}_2$-norm according to the four scenarios in Section~\ref{subsec:ex_H2}, depending on the noise level $\bar{d}$.
The different scenarios are 1) using only prior knowledge, 2) using only available data (equivalent to~\cite{waarde2022from}), 3) using prior knowledge and data, and 4) using exact model knowledge.}
\label{fig:example_H2}
\end{figure}
First, we note that it is possible to design a stabilizing controller with some guaranteed $\mathcal{H}_2$-norm for scenarios 1), 3), and 4) for all chosen noise levels.
The designs in scenario 1) and 4) are independent of $\bar{d}$ since they do not involve the measured data.
For the purely data-driven design in scenario 2), which is equivalent to the approach by~\cite{waarde2022from}, a stabilizing design is only possible for very small noise levels with $\bar{d}\leq0.02$.
Figure~\ref{fig:example_H2} indicates that the information provided by the data is more useful for design than the prior knowledge as long as $\bar{d}\leq0.02$.
In particular, if the disturbance bound $\bar{d}$ tends to zero, then the data provide an exact description of the uncertainty and both scenarios 2) and 3) match the performance $\gamma=1.93$ of scenario 4) with exact model knowledge.
On the other hand, for increasing noise levels, the prior knowledge is instrumental since the data alone do not allow for a stabilizing design.
The performance in scenario 3) combining prior knowledge and data matches that of scenario 2) for
$\bar{d}=0$ and that of scenario 1) for $\bar{d}=0.2$, which correspond to the extreme cases where the data either uniquely specify the true uncertainty or are of little use, respectively.
Notably, for any noise levels $\bar{d}\in[0,0.2]$, scenario 3) provides the best possible performance bounds for the chosen multiplier classes since it trades off the available information in a systematic fashion.

For comparison, we also compute a model-based optimal $\mathcal{H}_2$-controller for an identified model resulting from least-squares estimation.
To be precise, for the noise level $\bar{d}=0.1$ we use the data described above as well as the matrices $A$, $B$, $B_w$, $C_z$, $D_z$ to solve
\begin{align*}
    X_+=AX+BU+B_w\mathrm{diag}(\delta_1I_2,\Delta_2)(C_zX+D_zU)
\end{align*}
for $\delta_1$ and $\Delta_2$ in a least-squares sense while satisfying the constraints $\delta_1^2\leq0.1$ and $\Delta_2\Delta_2^\top\preceq0.5I$.
For the resulting system one can design an optimal
$\mathcal{H}_2$-controller with performance level $1.87$.
However, the $\mathcal{H}_2$-performance when applying this controller to the \emph{true} system is $1.93$.
Although one might argue that the least-squares-based controller outperforms all those in scenarios 1)-3) for this specific noise level, it does so with an overly high confidence.
This can be attributed to the fact that it does not take the uncertainty due to the noise into account and, hence, fails to come along with any theoretical guarantees.
In fact, the least-squares-based controller even fails to robustly stabilize~\eqref{eq:sys_prior} w.r.t.\ the prior uncertainty set $\Dep$.
Thus, without any further knowledge about the true plant, it is impossible to certify any properties for the closed loop.

\subsection{Comparison of different disturbance multipliers}\label{subsec:ex_H2_multipliers}

Let us now analyze the influence of the choice of multipliers used to describe the disturbance bound $\lVert\hat{d}_k\rVert_{\infty}\leq\bar{d}$ on the closed-loop performance for the same setting as in Section~\ref{subsec:ex_H2}.
We focus on scenario 3), the combined case with prior knowledge and data, for four classes of multipliers $\bm{P}_d$:
i) simple multipliers~\eqref{eq:noise_multiplier_quadratic} with $N=200$ as in Section~\ref{subsec:ex_H2}),
ii) convex hull multipliers~\eqref{eq:noise_multiplier_convex_hull} for which we use only $N=5$ data points for computational reasons, and
diagonal multipliers~\eqref{eq:noise_multiplier_diagonal}, where we translate $\lVert\hat{d}_k\rVert_{\infty}\leq\bar{d}$ into $\lVert\hat{d}_k\rVert_2\leq\sqrt{n_d}\bar{d}$ with iii) $N=5$ and iv) $N=20$ data points.
For each of these choices, the robust bounds on the closed-loop $\mathcal{H}_2$-norm based on Theorem~\ref{thm:robust_H2} are displayed in Figure~\ref{fig:example_H2_multipliers_noise}.

\begin{figure}
\begin{center}
\includegraphics[width=0.49\textwidth]{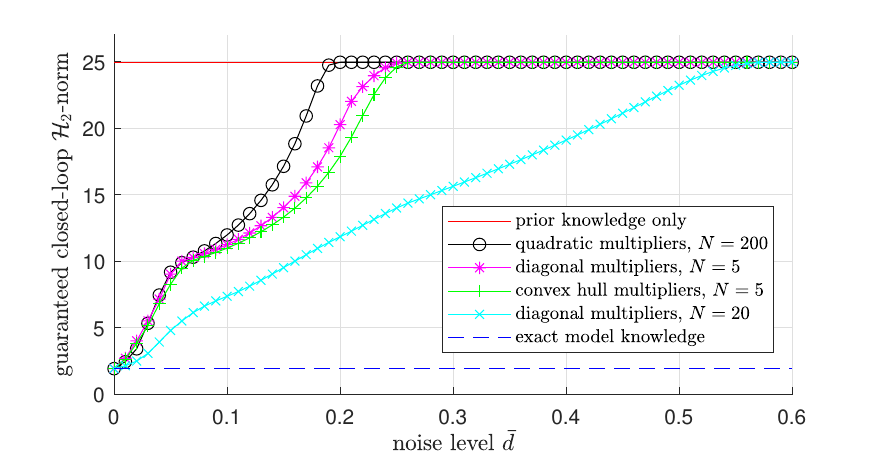}
\end{center}
\caption{Guaranteed closed-loop $\mathcal{H}_2$-norm when using prior knowledge and data (scenario 3) in Section~\ref{subsec:ex_H2}), depending on the noise level $\bar{d}$ and for different disturbance multipliers in Assumption~\ref{ass:noise_multipliers}.}
\label{fig:example_H2_multipliers_noise}
\end{figure}

First, we observe the same qualitative behavior for all multiplier classes:
The guaranteed performance levels approach that of the nominal $\mathcal{H}_2$-controller for small noise levels and that of the robust controller using only prior knowledge for large ones.
The convex hull and diagonal multipliers both have superior performance if compared to the simple ones in~\eqref{eq:noise_multiplier_quadratic} for most noise levels, even when using significantly fewer data points ($N=5$ instead of $N=200$).
The best performance is achieved when using diagonal multipliers with $N=20$ data points.
This indicates an advantage of diagonal multipliers over convex hull multipliers:
The reduced computational complexity allows us to employ longer data trajectories which, in turn, improves the resulting performance.
This advantage is not shared by the quadratic multipliers since the disturbance bound~\eqref{eq:noise_bound_quadratic}
is based on a conservative over-approximation of $\lVert\hat{d}_k\rVert_{\infty}\leq\bar{d}$ and, therefore, the resulting set of learnt uncertainties $\tDel$ does not necessarily shrink if additional data points are included.

To analyze the latter issue in more detail, let us perform the same experiments, but for a fixed
noise level $\bar{d}=0.15$ and varying the data length $N$.
We only consider quadratic~\eqref{eq:noise_multiplier_quadratic} and diagonal~\eqref{eq:noise_multiplier_diagonal} multipliers for reasons of computational complexity.
Figure~\ref{fig:example_H2_multipliers_data} shows that diagonal multipliers \emph{always} outperform the quadratic ones for the considered data lengths.
While the robust closed-loop $\mathcal{H}_2$-norm is non-increasing for longer data with diagonal multipliers, this is not the case for quadratic multipliers, for which the best possible performance is achieved with only $N=40$ data points.
This is an important insight:
A key advantage of data-driven control by~\cite{waarde2022from} in comparison to~\cite{persis2020formulas,berberich2020design} is that the size of the involved LMIs is independent of the data length and, thus, significantly longer data trajectories can be employed.
However, according to Figure~\ref{fig:example_H2_multipliers_data}, increasing the data length $N$ might even deteriorate the guaranteed performance when using the simple noise description~\eqref{eq:noise_bound_quadratic}.
We conclude that our more sophisticated description of disturbance bounds as in Assumption~\ref{ass:noise_multipliers} is instrumental in order to fully exploit the information provided by the available data.

\begin{figure}
\begin{center}
\includegraphics[width=0.49\textwidth]{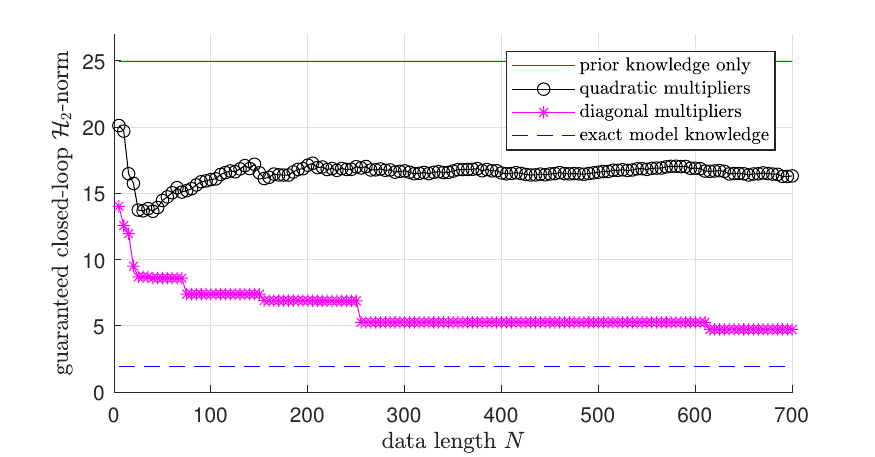}
\end{center}
\caption{Guaranteed closed-loop $\mathcal{H}_2$-norm when using prior knowledge and data (scenario 3) in Section~\ref{subsec:ex_H2}), depending on the data length $N$ and for different disturbance multipliers in Assumption~\ref{ass:noise_multipliers}.}
\label{fig:example_H2_multipliers_data}
\end{figure}

\subsection{Data-driven $\mathcal{H}_{\infty}$-loop-shaping}\label{subsec:ex_Hinf}
In a realistic robust controller design scenario, direct minimization of the $\mathcal{H}_{\infty}$-norm is rarely of practical use.
Instead, one typically includes weights to shape the frequency response of the closed loop.
Such weights constitute a natural form of prior knowledge, which can be handled by our framework as we demonstrate by revisiting Example~\ref{ex:introduction}.
We compute an exact discretization of~\eqref{eq:ex_satellite_LFT} with sampling time $0.05s$ leading to matrices $\bar{A}$, $\bar{B}$, $\bar{C}_z$, $\bar{D}_z$, $\bar{B}_w$, $\bar{B}_d$ of a discrete-time LFT as in~\eqref{eq:sys}.
Our goal is to design a controller which stabilizes the discretized system while rendering the influence of the disturbance $\tilde{d}$ on the deviation of the angle $\theta_2$ from zero and on the control input small.
More precisely, we want to achieve a possibly small $\mathcal{H}_{\infty}$-norm of the channel $d\mapsto\begin{bmatrix}w_1\theta_2\\w_2u\end{bmatrix}$, where $w_1$, $w_2$ are dynamic filters that allow us to trade off the two objectives of small tracking errors and control inputs.
We choose $w_1(z)$ as an exact discretization of the continuous-time low-pass filter $\tilde{w}_1(s)=\frac{0.5}{s+0.005}$
and $w_2(s)=0.1$.
If realizing the filter $w_1$ as
$x^f_{k+1}=A^f_1x^f_k+A^{f}_2x_k+B^fu_k$,  the combined dynamics of the discretized system and the filters 
take the form~\eqref{eq:sys} with
\begin{align*}
&A=\mleft[\begin{array}{c|c}\bar{A}&0\\\hline A^{f}_2&A^f_1
\end{array}\mright],\>\>
B=\mleft[\begin{array}{c}\bar{B}\\\hline B^f\end{array}\mright],\>\>B_w=\mleft[\begin{array}{c}\bar{B}_w\\\hline 0\end{array}\mright],\\
&B_d=\mleft[\begin{array}{c}\bar{B}_d\\\hline 0\end{array}\mright],\>\>
C_z=\mleft[\begin{array}{c|c}\bar{C}&0\end{array}\mright],\>\>D_z=\bar{D}.
\end{align*}
A performance channel according to the above specifications can be defined by choosing $C_e=\mleft[\begin{array}{c|c}0&1\\0&0\end{array}\mright]$, $D_{eu}=\begin{bmatrix}0\\0.1\end{bmatrix}$.
For the controller design, we do not consider any prior knowledge on the uncertainty $\Delta_{\tr}$, i.e., $\Dep=\mathbb{R}^{n_w\times n_z}$.
However, we generate data of length $N=100$ for the discretization of~\eqref{eq:ex_satellite_LFT} by sampling the input $u_k$ uniformly from $[-1,1]$ and
the disturbance $\tilde{d}_k$ with
$\lVert\tilde{d}\rVert_2\leq \bar{d}$ for $\bar{d}=5$.
This leads to the multipliers~\eqref{eq:noise_multiplier_quadratic} with $Q_d=-I$, $S_d=0$, $R_d=\bar{d}^2I$.
Based on Theorem~\ref{thm:robust_perf_nl}, we can now design a static state-feedback controller with performance specification
$Q_p=-\gamma^2I$, $S_p=0$ and $R_p=I$
for the discretized closed-loop system, which results in a guaranteed bound on the closed-loop $\mathcal{H}_{\infty}$-norm  of
$\gamma=0.22$ for the channel $d\mapsto e$.

The Bode plots of the continuous-time open-loop and closed-loop transfer functions are displayed in Figure~\ref{fig:bode} for frequencies below the Nyquist frequency.
The design specifications are clearly met, i.e., the depicted magnitude plots lie below those of the inverse filter dynamics $w_1$, $w_2$. Figure~\ref{fig:bode} also shows the results for a nominal design based on complete model knowledge, which closely match those with data-based synthesis.
To summarize, the proposed framework can exploit measured data to perform loop-shaping, which is a well-studied control problem leading to controllers with good closed-loop performance if the filters are chosen suitably~\cite{skogestad1996multivariable}.

\begin{figure}
\begin{center}
\includegraphics[width=0.5\textwidth]{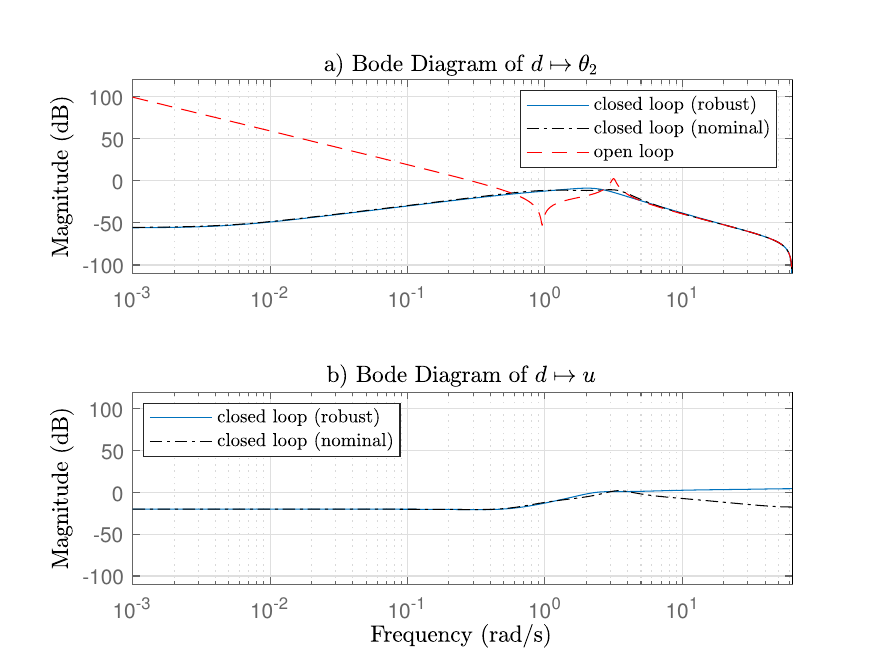}
\end{center}
\caption{Bode plot of the transfer function $d\mapsto e$ for the example in Section~\ref{subsec:ex_Hinf}.
Subfigure a) shows the magnitude of the channel $d\mapsto\theta_2$ in open ($K=0$) and closed loop and Subfigure b) shows the magnitude of the closed-loop channel $d\mapsto u$.
For the closed loop, both the robust design as well as the nominal case with exact model knowledge are displayed.}
\label{fig:bode}
\end{figure}

\section{Conclusion}\label{sec:conclusion}
In this paper, we presented a novel and flexible framework for systematically combining prior knowledge and measured data using robust control theory.
We showed how prior multipliers for a structured uncertainty can be translated into multipliers for a transformed full-block uncertainty, without losing information about the structure.
Next, we used noisy data in order to learn additional multipliers for this full-block uncertainty.
A further key novelty over previous works was the introduction of a general, multiplier-based disturbance description, which unifies and extends noise bounding techniques in the literature.
We then combined the derived multipliers to design controllers with robust stability and performance guarantees.
The generality of our framework allowed for seamless extensions to additional nonlinear uncertainties in the loop or to output-feedback design using input-output data.
Finally, we demonstrated the validity of the proposed approach with two examples, showcasing how simultaneously exploiting measured data and prior knowledge leads to superior performance if compared to purely data-driven approaches.


\bibliographystyle{IEEEtran}
\bibliography{Literature}

\section*{Appendix}
\renewcommand\thesection{\Alph{section}}
\setcounter{section}{0}
\section{Proof of Lemma~\ref{lem:prior_bounds}}\label{app:A}
\noindent\textbf{Proof of $\tDep\supseteq B_w\Dep$:}
Let $\Delta\in\Dep$ and $\tilde{P}\in\bm{\tilde{P}}$.
Then, there exist $P_j\in\bm{P}_j$ with
\begin{align*}
    \tilde{P}=\sum_{j=1}^{\ell}\begin{bmatrix}L_j^\top&0\\0&B_j^\top\end{bmatrix}^\top P_j
    \begin{bmatrix}L_j^\top&0\\0&B_j^\top\end{bmatrix}.
\end{align*}
By definition of $\Dep$, we have
\begin{align}\label{eq:lem_prior_bounds_proof1}
    \begin{bmatrix}\Delta_j^\top\\I\end{bmatrix}^\top
    P_j\begin{bmatrix}\Delta_j^\top\\I\end{bmatrix}\succeq0,\>\>j=1,\dots,\ell.
\end{align}
Now define $\tDe\coloneqq B_w\Delta$ such that
\begin{align}\label{eq:lem_prior_bounds_proof2}
    L_j^\top\tDe^\top=L_j^\top\Delta^\top B_w^\top\stackrel{\eqref{eq:Bw_Delta_Lj_equation}}{=}\Delta_j^\top B_j^\top.
\end{align}
By left- and right-multiplying~\eqref{eq:lem_prior_bounds_proof1} with $B_j$ and $B_j^\top$, respectively, we infer
\begin{align*}
    0&\preceq\begin{bmatrix}\Delta_j^\top B_j^\top\\B_j^\top\end{bmatrix}^\top P_j
    \begin{bmatrix}\Delta_j^\top B_j^\top\\B_j^\top\end{bmatrix}\\
    &\stackrel{\eqref{eq:lem_prior_bounds_proof2}}{=}\begin{bmatrix}\tDe^\top\\ I\end{bmatrix}^\top \begin{bmatrix}L_j^\top&0\\0&B_j^\top\end{bmatrix}^\top
    P_j\begin{bmatrix}L_j^\top&0\\0&B_j^\top\end{bmatrix}\begin{bmatrix}\tDe^\top\\I\end{bmatrix}.
\end{align*}
By summing up and since $\tilde{P}$ was arbitrary, we deduce
\begin{align*}
    \begin{bmatrix}\tDe^\top\\I\end{bmatrix}^\top \tilde{P}\begin{bmatrix}\tDe^\top\\I\end{bmatrix}\succeq0\>\>\text{for all}\>\>
    \tilde{P}\in\bm{\tilde{P}}.
\end{align*}
This implies $\tilde{\Delta}=B_w\Delta\in\tDep$ and hence, since $\Delta\in\Dep$ was arbitrary, $B_w\Dep\subseteq\tDep$.\\
\textbf{Proof of $\tDep\subseteq B_w\Dep$:}
Let $\tDe\in\tDep$, fix $j\in\{1,\dots,\ell\}$, and define $\tDe_j\coloneqq\tDe L_j$.
Further, choose $P_j\in\bm{P}_j$ as in Assumption~\ref{ass:prior_multipliers} such that its left-upper block $Q_j$ is negative definite.
Since $0\in\bm{P}_k$ for $k\neq j$, the definition of $\bm{\tilde{P}}$ implies
\begin{align}\label{eq:lem_prior_bounds_proof9}
    \tilde{P}_j\coloneqq\begin{bmatrix}L_j^\top&0\\0&B_j^\top\end{bmatrix}^\top P_j\begin{bmatrix}L_j^\top&0\\0&B_j^\top\end{bmatrix}\in\bm{\tilde{P}},
\end{align}
compare~\eqref{eq:prior_multiplier_transformed}.
Using $\tDe\in\tDep$, this leads to
\begin{align}\label{eq:lem_prior_bounds_proof3}
    0&\preceq\begin{bmatrix}\tDe^\top\\I\end{bmatrix}^\top\tilde{P}_j\begin{bmatrix}\tDe^\top\\I\end{bmatrix}\\\nonumber
    &=\begin{bmatrix}\tDe^\top\\I\end{bmatrix}^\top\begin{bmatrix}L_j^\top&0\\0&B_j^\top\end{bmatrix}^\top P_j
    \begin{bmatrix}L_j^\top&0\\0&B_j^\top\end{bmatrix}\begin{bmatrix}\tDe^\top\\I\end{bmatrix}\\\nonumber
    &=\begin{bmatrix}(\tDe L_j)^\top\\B_j^\top\end{bmatrix}^\top P_j\begin{bmatrix}(\tDe L_j)^\top\\B_j^\top\end{bmatrix}
    =\begin{bmatrix}\tDe^\top_j\\B_j^\top\end{bmatrix}^\top P_j\begin{bmatrix}\tDe_j^\top\\B_j^\top\end{bmatrix}.
\end{align}
Thus, for any $v\in\mathbb{R}^n$ with $B_j^\top v=0$, we have
\begin{align*}
    0\preceq v^\top\begin{bmatrix}\tDe_j^\top\\B_j^\top\end{bmatrix}^\top P_j\begin{bmatrix}\tDe_j^\top\\B_j^\top\end{bmatrix}v=
    v^\top\tDe_j Q_j\tDe_j^\top v\stackrel{Q_j\prec0}{<}0
\end{align*}
if $\tDe_j^\top v\neq0$.
This guarantees that $\tDe_j^\top v=0$.
In conclusion, there exists some $\Delta_j$ with $\tDe_j=B_j\Delta_j$.
Let now $P_j\in\bm{P}_j$ be arbitrary and note that~\eqref{eq:lem_prior_bounds_proof9} still holds.
Then, as in~\eqref{eq:lem_prior_bounds_proof3},
\begin{align*}
    0&\preceq\begin{bmatrix}\tDe_j^\top\\B_j^\top\end{bmatrix}^\top P_j\begin{bmatrix}\tDe_j^\top\\B_j^\top\end{bmatrix}
    =B_j\begin{bmatrix}\Delta_j^\top\\I\end{bmatrix}^\top P_j\begin{bmatrix}\Delta_j^\top\\I\end{bmatrix}B_j^\top.
\end{align*}
By using the fact that $B_j$ has full column rank, we infer $\begin{bmatrix}\Delta_j^\top\\I\end{bmatrix}^\top
    P_j\begin{bmatrix}\Delta_j^\top\\I\end{bmatrix}\succeq0$, i.e., $\Delta_j\in\bm{\Delta}_j$.
Moreover, since $j$ was arbitrary, there exists $\Delta=\mathrm{diag}_{j=1}^{\ell}(\Delta_j)\in\Dep$ with $B_w\Delta=\begin{bmatrix}
B_1\Delta_1&\dots&B_{\ell}\Delta_{\ell}\end{bmatrix}=\begin{bmatrix}\tDe_1&\dots&\tDe_{\ell}\end{bmatrix}=\tDe$.
This shows $\tDe\in B_w\Dep$ and, thus, $\tDep\subseteq B_w\Dep$.
\qed

\section{Proof of Lemma~\ref{lem:data_bound}}\label{app:B}
Defining $\bm{\tilde{D}}\coloneqq B_d\bm{D}$, $\Delta\in\Del$ is equivalent to
\begin{align}\label{eq:data_bound2}
    M=\tilde{D}+B_w\Delta Z\>\>\text{for some}\>\>\tilde{D}\in\bm{\tilde{D}}.
\end{align}
With $\bm{\tilde{P}}_d\coloneqq\begin{bmatrix}I&0\\0&B_d^\top\end{bmatrix}^\top\!\!\! \bm{P}_d
    \begin{bmatrix}I&0\\0&B_d^\top\end{bmatrix}$
and using negative definiteness of the left-upper block of some $P_d\in\bm{P}_d$ (Assumption~\ref{ass:noise_multipliers}), it can be shown analogously to Lemma~\ref{lem:prior_bounds} that
\begin{align}\label{eq:tilde_D_def}
    \bm{\tilde{D}}=\left\{\tilde{D}\Bigm|
    \begin{bmatrix}\tilde{D}^\top\\I\end{bmatrix}^\top\tilde{P}_d
    \begin{bmatrix}\tilde{D}^\top\\I\end{bmatrix}\succeq0\>\>\text{for all}\>\>\tilde{P}_d\in\bm{\tilde{P}}_d\right\}.
\end{align}
Defining $\tilde{\Delta}\coloneqq B_w\Delta$,~\eqref{eq:data_bound2} is equivalent to $M-\tilde{\Delta} Z\in\bm{\tilde{D}}$.
For any $\tilde{P}_d\in\bm{\tilde{P}}_d$, we can hence infer
\begin{align*}
    0&\preceq\begin{bmatrix}M^\top-Z^\top\tilde{\Delta}^\top\\I\end{bmatrix}^\top
    \tilde{P}_d\begin{bmatrix}M^\top-Z^\top\tilde{\Delta}^\top\\I\end{bmatrix}\\
    &=\begin{bmatrix}\tilde{\Delta}^\top\\I\end{bmatrix}^\top
    \begin{bmatrix}-Z^\top&M^\top\\0&I\end{bmatrix}^\top
    \tilde{P}_d\begin{bmatrix}-Z^\top&M^\top\\0&I\end{bmatrix}
    \begin{bmatrix}\tilde{\Delta}^\top\\I\end{bmatrix},
\end{align*}
which proves that $B_w\Delta=\tilde{\Delta}\in\tDel$.

\section{Proof of Lemma~\ref{lem:combined_multipliers}}\label{app:C}
\noindent
\textbf{Proof of $\tDecom\supseteq B_w\Decom$:}
This follows directly from $\Decom=\Dep\cap\Del$ together with $B_w\Dep\subseteq\tDep$ and $B_w\Del\subseteq\tDel$.\\
\textbf{Proof of $\tDecom\subseteq B_w\Decom$:}
Suppose $\tilde{\Delta}\in\tDecom$, i.e., $\tilde{\Delta}\in\tDep$ and $\tilde{\Delta}\in\tDel$.
Then, according to Lemma~\ref{lem:prior_bounds}, we have $\tilde{\Delta}\in B_w\Dep$ such that there exists $\Delta\in\Dep$ with $\tilde{\Delta}=B_w\Delta$.
For an arbitrary $\tilde{P}_d\in\bm{\tilde{P}}_d$, we then infer
\begin{align*}
    \begin{bmatrix}M^\top-Z^\top\tilde{\Delta}^\top\\I\end{bmatrix}^\top
    \tilde{P}_d\begin{bmatrix}M^\top-Z^\top\tilde{\Delta}^\top\\I\end{bmatrix}\succeq0,
\end{align*}
i.e., $\tilde{D}\coloneqq M-\tilde{\Delta}Z$ satisfies $\tilde{D}\in\bm{\tilde{D}}$.
This, in turn, implies the existence of $D\in\bm{D}$ such that $M=B_dD+\tilde{\Delta}Z$.
Using $\tilde{\Delta}=B_w\Delta$, we conclude $\Delta\in\Del$ and, hence, $\tilde{\Delta}\in B_w\Decom$.
\qed

%
\begin{IEEEbiography}[{\includegraphics[width=1in,height=1.25in,clip,keepaspectratio]{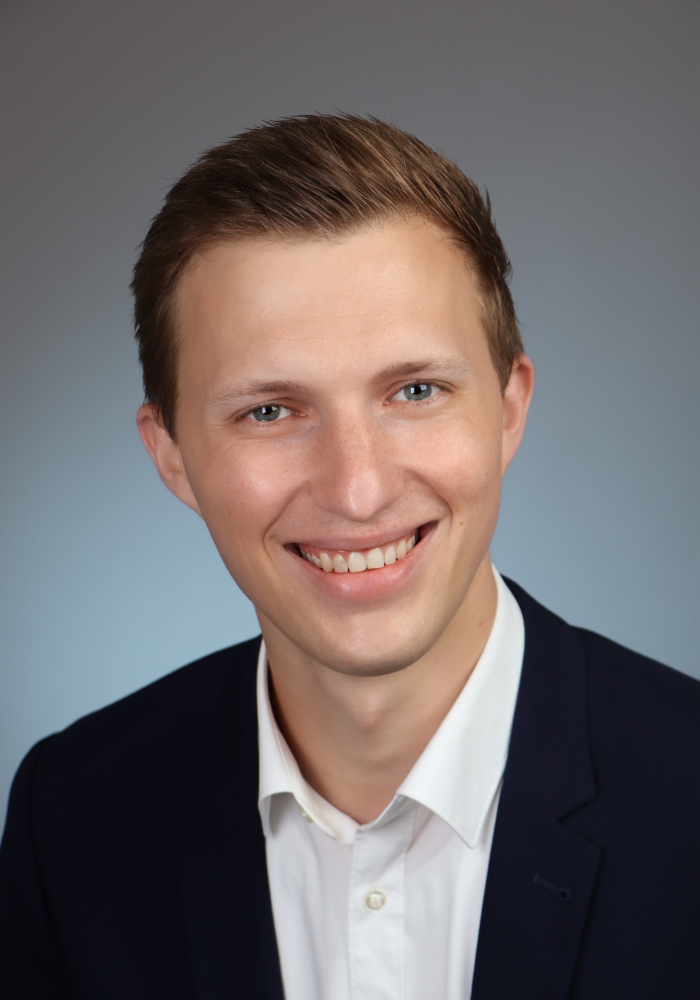}}]{Julian Berberich}
received the Master's degree in Engineering Cybernetics from the University of Stuttgart, Germany, in 2018. Since 2018, he has been a Ph.D. student at the Institute for Systems Theory and Automatic Control under supervision of Prof. Frank Allg\"ower and a member of the International Max-Planck Research School (IMPRS) at the University of Stuttgart. He has received the Outstanding Student Paper Award at the 59th Conference on Decision and Control in 2020. His research interests are in the area of data-driven analysis and control.
\end{IEEEbiography}

\begin{IEEEbiography}[{\includegraphics[width=1in,height=1.25in,clip,keepaspectratio]{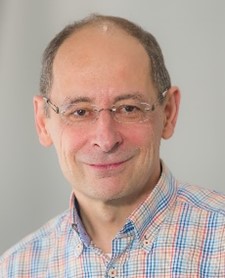}}]
{Carsten W. Scherer}
received his Ph.D. degree in mathematics from the University of W\"urzburg (Germany) in 1991. In 1993, he joined Delft University of Technology, The Netherlands, where he held positions as an assistant, associate and full professor at the Delft Center for Systems and Control. Since March 2010 he holds the SimTech Chair for Mathematical Systems Theory at the Department of Mathematics, University of Stuttgart, Germany.

Dr. Scherer acted as the chair of the IFAC technical committee on Robust Control, and served as an AE for the IEEE Transactions on Automatic Control, Automatica, Systems \& Control Letters and the European Journal of Control. Since 2013 he is an IEEE fellow “for contributions to optimization-based robust controller synthesis“.

Dr. Scherer’s main research activities cover various topics in applying optimization techniques for developing new advanced controller design algorithms and their application to mechatronics and aerospace systems.
\end{IEEEbiography}

\begin{IEEEbiography}[{\includegraphics[width=1in,height=1.25in,clip,keepaspectratio]{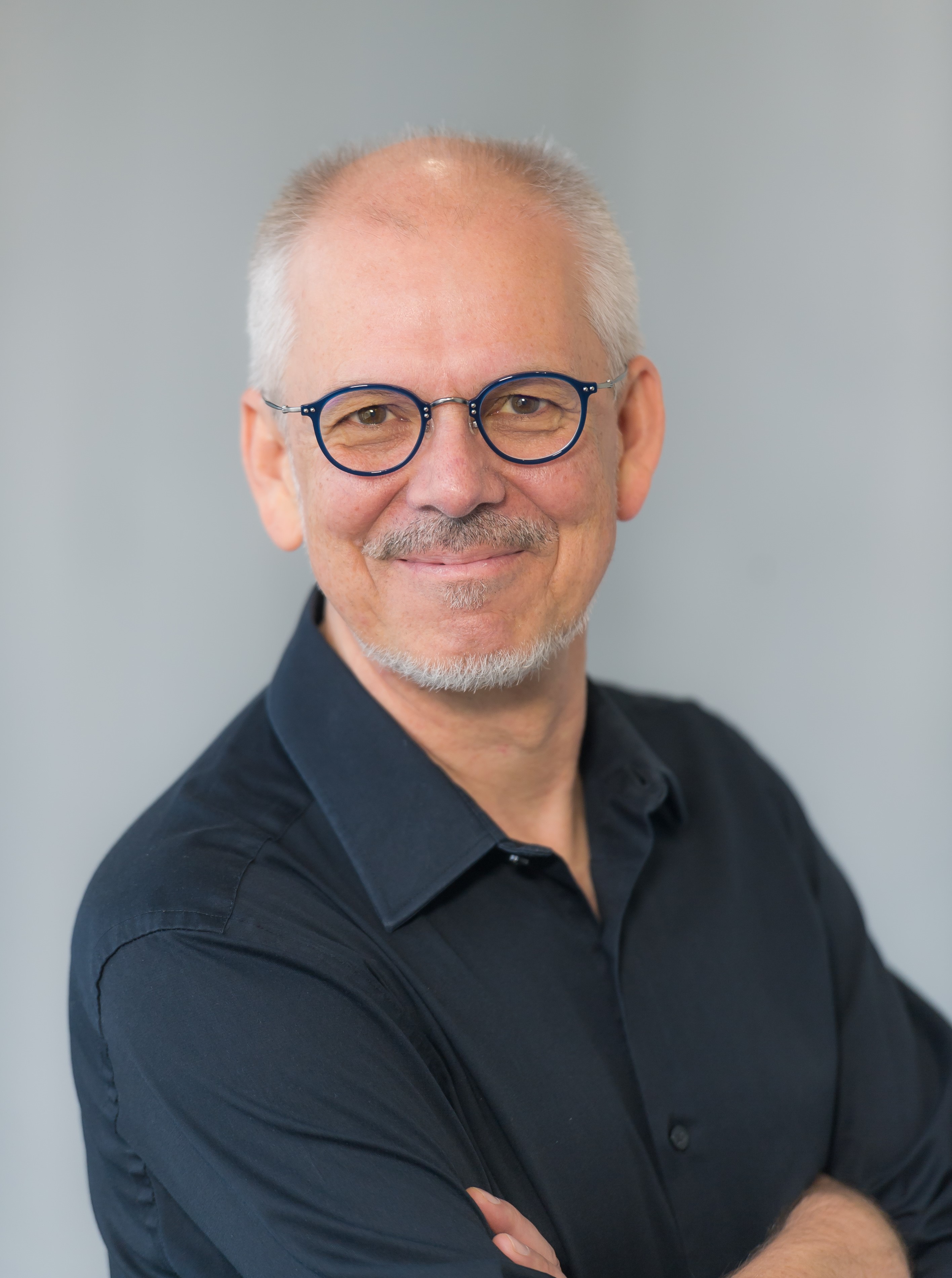}}]{Frank Allg\"ower}
is professor of mechanical engineering at the University of Stuttgart, Germany, and Director of the Institute for Systems Theory and Automatic Control (IST) there.\\
Frank is active in serving the community in several roles: Among others he has been President of the International Federation of Automatic Control (IFAC) for the years 2017-2020, Vice-president for Technical Activities of the IEEE Control Systems Society for 2013/14, and Editor of the journal Automatica from 2001 until 2015. From 2012 until 2020 Frank served in addition as Vice-president for the German Research Foundation (DFG), which is Germany’s most important research funding organization. \\
His research interests include predictive control, data-based control, networked control, cooperative control, and nonlinear control with application to a wide range of fields including systems biology.
\end{IEEEbiography}







\end{document}